\documentclass[american,aps,prl,reprint,longbibliography]{revtex4-2}
\usepackage[T1]{fontenc}
\usepackage[utf8]{inputenc}
\usepackage{color}
\usepackage{babel}
\usepackage{amsmath}
\usepackage{amsthm}
\usepackage{amssymb}
\usepackage[pdfusetitle,
 bookmarks=true,bookmarksnumbered=false,bookmarksopen=false,
 breaklinks=false,pdfborder={0 0 0},pdfborderstyle={},backref=false,colorlinks=true]
 {hyperref}
\hypersetup{
 allcolors=magenta}

\makeatletter

\theoremstyle{plain}
\newtheorem{thm}{\protect\theoremname}

\newtheorem{proposition}{Proposition}

\usepackage{times}
\usepackage{txfonts}
\usepackage{braket}
\usepackage{colortbl}
\usepackage{graphicx} 

\makeatother

\providecommand{\theoremname}{Theorem}

\begin{document}
\title{Robustness of quantum data hiding against entangled catalysts and memory}

\author{Aby Philip}
\author{Alexander Streltsov}
\email{streltsov.physics@gmail.com}
\affiliation{Institute of Fundamental Technological Research, Polish Academy of Sciences, Pawi\'{n}skiego 5B, 02-106 Warsaw, Poland}

\begin{abstract}
Quantum data hiding stores classical information in bipartite quantum states that are, in principle, perfectly distinguishable, yet remain almost indistinguishable without access to a quantum communication channel. Here, we investigate whether this limitation can be overcome when the communicating parties are assisted by additional quantum resources. We develop a general framework for state discrimination that unifies catalytic and memory-assisted local discrimination protocols and analyze their power to reveal hidden information. We prove that when the hiding states are separable, neither entangled catalysts nor quantum memory can increase the optimal discrimination probability, establishing the robustness of separable data-hiding schemes. In contrast, for some entangled states, a reusable quantum memory turns locally indistinguishable states into ones that can be discriminated almost perfectly. Our results delineate the fundamental limits of catalytic and memory-assisted state discrimination and identify separable encodings as a robust strategy for quantum data hiding. 
\end{abstract}

\maketitle

\section{Introduction}

Quantum entanglement, the hallmark nonclassical correlation of quantum mechanics~\cite{Horodecki_2009}, is a key resource underlying much of quantum information science~\cite{PhysRevLett.70.1895,PhysRevLett.69.2881,Shor_1997,PhysRevLett.67.661}. 
Shared entanglement enables protocols with no classical analog, such as quantum teleportation~\cite{PhysRevLett.70.1895} and quantum super-dense coding~\cite{PhysRevLett.69.2881}, and enhances the performance of tasks in computation~\cite{Shor_1997} and cryptography~\cite{PhysRevLett.67.661}. 
However, not all genuinely quantum phenomena rely on entanglement, and in some cases, excessive entanglement may even hinder information processing. 
Notable examples include \emph{nonlocality without entanglement}, where non-entangled states exhibit intrinsically nonclassical behavior~\cite{PhysRevA.59.1070}, and the observation that certain highly entangled states are too entangled to serve as useful resources for quantum computation~\cite{PhysRevLett.102.190501}.

A foundational problem, dating to the earliest days of quantum information science, is local state discrimination~\cite{PhysRevLett.66.1119}: given a set of states shared between distant parties, Alice and Bob, can they identify which state they hold using only local operations and classical communication (LOCC)? For two orthogonal pure states, the answer is always yes~\cite{PhysRevLett.85.4972}. However, when considering larger and more general sets of states, striking effects emerge—most notably, the aforementioned phenomenon of nonlocality without entanglement, where orthogonal product states cannot be perfectly distinguished by LOCC alone~\cite{PhysRevA.59.1070}. 
This phenomenon has led to the discovery of families of bipartite~\cite{Berry_Groisman_2001,Yu_2012,yang2015characterizing,xu2016locally} and multipartite~\cite{divincenzo2003unextendible,4957660,PhysRevA.74.052103,PhysRevA.88.024301,PhysRevA.98.022303,PhysRevA.93.032341,PhysRevA.95.052344,PhysRevA.105.032407} orthogonal states that remain locally indistinguishable. Such phenomena are crucial because they point out an operational gap between local and general quantum measurements. An important application is quantum data hiding~\cite{PhysRevLett.86.5807,PhysRevLett.89.097905,DiVincenzo_2002,Hayden_2004,Aubrun_2015}, where classical information is encoded into bipartite states that are perfectly distinguishable in principle, yet almost indistinguishable for any LOCC procedure, thereby furnishing an information-theoretic primitive for secret sharing.

Entanglement catalysis offers a way to overcome some limits of local state discrimination~\cite{Yu_2012,PhysRevA.105.032407}. In this setting, Alice and Bob may borrow an ancillary entangled state, a catalyst, that can interact with their systems during the protocol but must return exactly to its initial state~\cite{PhysRevLett.83.3566}. Remarkably, quantum catalysts can activate local distinguishability: there exist entangled states that are not perfectly distinguishable by LOCC alone yet become perfectly distinguishable in the presence of a suitable catalyst~\cite{Yu_2012,PhysRevA.105.032407}. Beyond discrimination, catalysis enlarges the scope of LOCC state transformations~\cite{PhysRevLett.83.3566,Neven_2021,PhysRevLett.127.150503,PhysRevLett.127.080502} and has been investigated across quantum thermodynamics~\cite{PhysRevLett.126.150502,PhysRevLett.132.200201,PhysRevLett.132.180202} and other quantum resource theories~\cite{Datta_2023,RevModPhys.96.025005}. 

Another quantum resource considered useful for information processing is the \emph{quantum memory}, which has been shown to be more powerful than its classical counterpart~\cite{Konig_2005,Berta_2010}. 
The use of quantum memory has been studied in the context of quantum networks~\cite{PhysRevA.80.022339} and quantum channel discrimination~\cite{10.1145/1250790.1250873,PhysRevA.81.032339,PhysRevLett.101.180501}.
In the setting considered here, a quantum memory refers to an auxiliary quantum system shared between Alice and Bob that can interact with their systems during the protocol. 
Unlike a catalyst, however, the quantum memory is not required to return to its initial state and can instead be reused in subsequent rounds of the protocol.

In this work, we investigate the role of entanglement catalysis and quantum memory in local state discrimination and data hiding. We prove that for any pair of separable states, access to either a catalyst or a quantum memory does not enhance the optimal discrimination probability achievable by LOCC. 
We further show that certain data hiding schemes, which are secure under standard LOCC protocols, become vulnerable once the parties are equipped with a reusable quantum memory: in this setting, memory assistance enables local discrimination with success probability arbitrarily close to unity. These results delineate when catalytic or memory-assisted protocols can and cannot overcome the fundamental limitations of quantum data hiding.

\section{Local state discrimination with entangled catalysts and memory}
Quantum state discrimination can be viewed as a game in which a referee prepares one of two possible quantum states, \(\{\rho_0, \rho_1\}\), and sends it to an agent whose objective is to determine which state was prepared by performing an appropriate quantum measurement on the received system~\cite{helstrom1969quantum,HELSTROM1967254}. If the two states are prepared with equal prior probability, the maximal success probability achievable by the agent in this task is given by~\cite{helstrom1969quantum,HELSTROM1967254}
\begin{equation} \label{eq:Popt}
P_{\mathrm{opt}}(\rho_{0}, \rho_{1})
= \frac{1}{2} + \frac{1}{4} \bigl\Vert \rho_{0} - \rho_{1} \bigr\Vert_{1},
\end{equation}
where \(\Vert M \Vert_{1} = \mathrm{Tr}\!\sqrt{M^{\dagger} M}\) is the trace norm.

In the previous setting, the agent had access to all quantum measurements allowed by quantum mechanics. 
An important variation of this scenario involves two spatially separated agents, Alice and Bob, who are then only allowed to implement operations via LOCC. We shall refer to this scenario as \emph{local state discrimination}. In this case, the referee prepares one of two bipartite quantum states, \(\rho_0^{AB}\) or \(\rho_1^{AB}\), and the subsystems $A$ and $B$ are given to Alice and Bob, respectively. 
When these states are prepared with equal prior probability, the optimal success probability for Alice and Bob to distinguish them with LOCC is given by~\cite{Matthews_2009}
\begin{equation} \label{eq:Plocc}
P_{\mathrm{LOCC}}(\rho_{0}, \rho_{1})
= \frac{1}{2} + \frac{1}{4} \bigl\Vert \rho_{0} - \rho_{1} \bigr\Vert_{\mathrm{LOCC}}.
\end{equation}
Here \(\Vert \cdot \Vert_{\mathrm{LOCC}}\) is the LOCC norm, we refer to the Methods section for a formal definition and more details. 

Quantum data hiding is a surprising phenomenon, implying the existence of state pairs $\rho_0$ and $\rho_1$ which are perfectly distinguishable in principle, but when distributed to two spatially separated parties, they become almost indistinguishable via LOCC~\cite{PhysRevLett.86.5807,PhysRevLett.89.097905,DiVincenzo_2002,Hayden_2004,Aubrun_2015}. Specifically, for any \(\varepsilon > 0\), there exist bipartite quantum states \(\rho_0^{AB}\) and \(\rho_1^{AB}\) such that~\cite{PhysRevLett.89.097905,Hayden_2004,Aubrun_2015,Ha_2025,mele2025optimisingquantumdatahiding}
\begin{align}
P_{\mathrm{opt}}(\rho_{0}, \rho_{1}) &= 1, \\
P_{\mathrm{LOCC}}(\rho_{0}, \rho_{1}) &< \frac{1}{2} + \varepsilon. \label{eq:DataHiding}
\end{align}
Although these states can be perfectly distinguished by a global measurement or, equivalently, when Alice and Bob have access to a quantum communication channel, they remain almost indistinguishable when restricted to LOCC.

In this work, we investigate local state discrimination and data hiding under more general strategies. 
A natural extension of the standard setting is to allow Alice and Bob to employ \emph{entangled catalysts}. 
In this scenario, Alice and Bob have access to an additional ancillary system \(A'B'\), referred to as the catalyst, which must be returned unchanged at the end of the process~\cite{PhysRevLett.83.3566}. 
Specifically, the goal is to find a quantum state of the catalyst \(\tau^{A'B'}\) and an LOCC protocol \(\Lambda_{\mathrm{LOCC}}\) such that
\begin{align}
\Lambda_{\mathrm{LOCC}}\!\left(\rho_{i}^{AB} \otimes \tau^{A'B'}\right) &= \sigma_{i}^{AB} \otimes \tau^{A'B'}, \label{eq:Catalyst-1}\\
P_{\mathrm{LOCC}}\!\left(\sigma_{0}^{AB}, \sigma_{1}^{AB}\right) &> P_{\mathrm{LOCC}}\!\left(\rho_{0}^{AB}, \rho_{1}^{AB}\right).\label{eq:Catalyst-2}
\end{align}
In other words, the states \(\rho_0^{AB}\) and \(\rho_1^{AB}\) are catalytically transformed into \(\sigma_0^{AB}\) and \(\sigma_1^{AB}\), respectively, such that the new pair can be better distinguished via LOCC, while the catalyst state \(\tau^{A'B'}\) remains unchanged. While catalytic transformations are known to enhance certain local state discrimination protocols~\cite{Yu_2012,PhysRevA.105.032407}, their potential impact on quantum data hiding remains largely unexplored.

To address this question, we introduce a general state discrimination framework that unifies all scenarios considered in this work, including the catalytic case. For this, let \( Z_j \) be an independent identically distributed (i.i.d.)\ random variable taking values in \(\{0,1\}\) uniformly at random. 
In the \(j\)-th round of the discrimination procedure, the state to be distinguished is given by \(\rho_{Z_j}^{AB}\). 
Let \(Y_j \in \{0,1\}\) denote the outcome of Alice and Bob's guess in the \(j\)-th round. We further define the variable \(X_j\) to represent whether the guess in round \(j\) is correct:
\begin{equation} \label{eq:X}
X_j = 
\begin{cases}
1, & \text{if } Y_j = Z_j, \\
0, & \text{otherwise.}
\end{cases}
\end{equation}
The total number of correct guesses after \(n\) rounds is then given by
\begin{equation} \label{eq:S}
S_n = \sum_{j=1}^n X_j.
\end{equation}
These definitions are completely general and can, in principle, be applied to any state discrimination procedure.

\begin{figure}
\includegraphics[width=0.9\columnwidth]{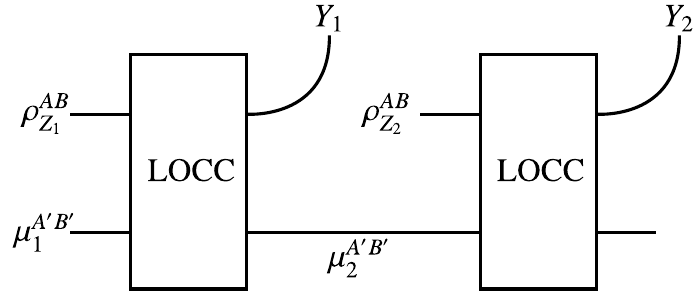}

\caption{\label{fig:1}\textbf{Local state discrimination with entanglement catalysis and quantum memory.} In each round, one of two quantum states \(\rho_{Z_1}^{AB}\) is sent to Alice and Bob, where \(Z_1\) is a random variable taking the values \(0\) or \(1\) with equal probability. 
In addition, Alice and Bob share a quantum memory \(A'B'\), initialized in the state \(\mu_1^{A'B'}\). 
They attempt to infer the value of \(Z_1\) by performing an LOCC protocol on the joint system \(\rho_{Z_1}^{AB} \otimes \mu_1^{A'B'}\) and recording their guess as \(Y_1 \in \{0,1\}\). 
In the next round, the updated memory state \(\mu_2^{A'B'}\) is reused to guess \(Z_2\), and the procedure continues iteratively. 
In the catalytic setting, the memory state remains unchanged throughout the process, that is, \(\mu_1^{A'B'} = \mu_2^{A'B'} = \mu_j^{A'B'}\) for all rounds \(j\). The figure shows the first two rounds of the process.
}

\end{figure}

Equipped with these tools, we can make precise the notion of an \emph{achievable success rate} for a general discrimination procedure. 
We say that the rate \(r \in [0,1]\) is achievable if, for every \(\varepsilon > 0\) and \(m > 0\), there exists \(n \ge m\) such that 
\begin{equation} \label{eq:P}
\operatorname{Pr}\!\left( S_n \ge r n \right) \ge 1 - \varepsilon .
\end{equation}
This definition means that there exist infinitely many values of $n$ for which, with probability arbitrarily close to one, the empirical success fraction \(S_n/n\) is no smaller than $r$.
Equivalently, there exists an unbounded and monotonically increasing integer sequence $\{n_k\}$ such that $\lim_{k\rightarrow\infty}\operatorname{Pr}\left(S_{n_{k}}\geq rn_{k}\right)=1$. The \emph{optimal success rate} $R$ is then obtained by taking supremum over all achievable rates $r$. 

The general framework introduced above naturally encompasses both the standard state discrimination setting discussed at the beginning of this section and local state discrimination discussed throughout this work. 
In the standard scenario, one recovers Eq.~(\ref{eq:Popt}) as the expression for the optimal success rate, while applying the framework to the LOCC setting yields \(P_{\mathrm{LOCC}}\) as defined in Eq.~(\ref{eq:Plocc}).

We will apply the general framework to investigate local state discrimination and data hiding in the catalytic setting, and further extend our analysis to more general transformation protocols by relaxing the catalytic constraint to include the use of a \emph{quantum memory}. In this scenario, the ancillary system \(A'B'\), which now serves as the quantum memory, may evolve and take different states throughout the process, see Fig.~\ref{fig:1}.
The total initial state is given by \(\rho_{Z_1}^{AB} \otimes \mu_1^{A'B'}\), where \(\mu_1^{A'B'}\) denotes the initial memory state and \(Z_1 \in \{0,1\}\) labels the prepared input. 
Alice and Bob apply an LOCC protocol to this composite system and, using the resulting classical data, produce a guess \(Y_1\). 
The final memory state, denoted \(\mu_2^{A'B'}\), depends on the specific protocol and on the corresponding measurement outcomes and, in general, does not coincide with the initial state \(\mu_1^{A'B'}\). In the second round, the updated memory state \(\mu_2^{A'B'}\) is reused: Alice and Bob apply an LOCC protocol to the state $\rho_{Z_2}^{AB} \otimes \mu_2^{A'B'}$ to produce the next guess \(Y_2\). 
Proceeding iteratively, the quantum memory is updated round by round and fed back into the protocol. Note that, in this setting, the sequence \(\{Y_j\}\) is, in general, not i.i.d.

With the main concepts in place, we present our key results in the following section.

\section{Main results}

The central question of this work is whether entangled catalysts or quantum memory can overcome data hiding; 
that is, whether there exist states \(\rho_0\) and \(\rho_1\) satisfying Eq.~(\ref{eq:DataHiding}) that nevertheless become perfectly distinguishable under LOCC when assisted by an entangled catalyst returned unchanged, or a reusable quantum memory. 

This question is particularly relevant from a security perspective, especially when data hiding is employed to protect information. 
If Alice and Bob can recover the hidden data with the aid of a catalyst, the encoding scheme can no longer be regarded as secure, since the parties are able to reveal the information without consuming any entanglement in the process.

The next theorem settles this question for the case where \(\rho_0\) and \(\rho_1\) are separable. Here, \(R_{\mathrm{c}}\) and \(R_{\mathrm{m}}\) denote the optimal success rates achievable with, respectively, an entangled catalyst and a quantum memory.

\begin{thm} \label{thm:Separable}
    For separable data hiding states $\rho_{0,1}$ neither quantum catalysis nor quantum memory can increase the optimal success probability:
    \begin{equation}
    R_{\mathrm{m}}(\rho_{0},\rho_{1})=R_{\mathrm{c}}(\rho_{0},\rho_{1})=P_{\mathrm{LOCC}}(\rho_{0},\rho_{1}).
    \end{equation}
\end{thm}
\noindent We outline the main idea of the proof for the catalytic setting here, while a detailed derivation, including the case of a quantum memory, is provided in the Methods section. 

The proof proceeds by contradiction. 
Assume that there exists a catalytic procedure such that \(R_{\mathrm{c}}(\rho_{0},\rho_{1}) > P_{\mathrm{LOCC}}(\rho_{0},\rho_{1})\). 
Let \(\{X_j^{\mathrm{c}}\}\) denote the i.i.d. sequence of variables labeling the correct guesses made by Alice and Bob with the aid of the catalyst (see also Eq.~(\ref{eq:X})).
Analogously to Eq.~(\ref{eq:S}), we define $S_n^{\mathrm{c}} = \sum_{j=1}^n X_j^{\mathrm{c}}$ to be the number of correct guesses after $n$ iterations. Since, by assumption, the catalytic procedure outperforms the one without catalysts, Hoeffding’s inequality implies that for some \(q > P_{\mathrm{LOCC}}(\rho_{0}, \rho_{1})\), the following bound holds for all \(\delta > 0\) and all \(n\):
\begin{equation}
\operatorname{Pr}\!\left(\left| \frac{S_n^{\mathrm{c}}}{n} - q \right| \le \delta \right) 
\ge 1 - 2 e^{-2 n \delta^2}.
\end{equation}
In other words, if the catalyst indeed provides an advantage, then, with probability arbitrarily close to one, the empirical success fraction \(S_n^{\mathrm{c}}/n\) must exceed \(P_{\mathrm{LOCC}}\) for sufficiently large~\(n\).

An important consequence is that Alice and Bob could exploit such an advantage to probe the state of the catalyst itself. 
Since, by assumption, the states \(\rho_0\) and \(\rho_1\) are separable, each round of the discrimination procedure can also be regarded as an LOCC protocol acting on the catalyst, producing a sequence of classical outcomes \(\{Y_j\}\). 
Crucially, the collected data \(\{Y_j\}\) could then be used to perfectly determine whether the catalyst was initialized in the intended entangled state \(\tau^{A'B'}\) or in some separable state \(\nu_{\mathrm{sep}}^{A'B'}\).
Since Alice and Bob are provided with only one copy of the catalyst, this would contradict the fundamental bound on quantum state discrimination given in Eq.~(\ref{eq:Popt}). 
By the same reasoning, the argument extends to the more general scenario in which Alice and Bob employ a reusable quantum memory.

It is known that quantum catalysts can, in general, enhance the success probability of local state discrimination~\cite{Yu_2012,PhysRevA.105.032407}. 
However, Theorem~\ref{thm:Separable} shows that no such enhancement is  possible when the states involved are separable. 
These observations lead to the question: when, in general, can catalysts or quantum memory overcome quantum data hiding. Specifically, when can Alice and Bob exploit catalysts or quantum memory to perfectly distinguish two states that are otherwise almost indistinguishable under LOCC. 
The following theorem provides an answer to this question for the case of a reusable quantum memory.

\begin{thm} \label{thm:2}
For every $\varepsilon,\delta > 0$ there exist data hiding states $\rho_0$, $\rho_1$ such that 
\begin{align}
P_{\mathrm{LOCC}}(\rho_{0},\rho_{1}) & <\frac{1}{2}+\varepsilon, \label{eq:Theorem2-1}\\
R_{\mathrm{m}}(\rho_{0},\rho_{1}) & >1-\delta. \label{eq:Theorem2-2}
\end{align}
\end{thm}

\noindent We present the main idea of the proof below and refer to the Methods section for further details.

To prove Theorem~\ref{thm:2}, we start by constructing two orthogonal states \(\rho_0\) and \(\rho_1\) that satisfy Eq.~(\ref{eq:Theorem2-1}). Note that if the states are orthogonal, Alice and Bob could, in principle, distinguish them perfectly using a sufficiently large number of shared Bell states. 
We provide a construction of states \(\rho_0\) and \(\rho_1\) such that Alice and Bob can distinguish them with \(k\) Bell states while simultaneously obtaining, in each round of the discrimination procedure, a pure state $\ket{\psi}$ with entanglement entropy exceeding \(k\). 
Alice and Bob can store sufficiently many copies of the pure state $\ket{\psi}$ in the quantum memory and subsequently distill them into Bell states in later rounds. 

Using this construction, we demonstrate that the proposed discrimination procedure achieves the claimed performance. In more detail, let \(X_j^{\mathrm{m}}\) denote the variable indicating a correct guess in the \(j\)-th round in the presence of a quantum memory, and define \(S_n^{\mathrm{m}} = \sum_{j=1}^n X_j^{\mathrm{m}}\) analogously to Eq.~(\ref{eq:S}). 
As we prove in the Methods section using the procedure described above, for any \(r < 1\), there exists an initial state of the quantum memory \(\nu^{A'B'}\) and an LOCC protocol such that, for every \(\varepsilon, m > 0\), one can find \(n \ge m\) satisfying
\begin{equation}
    \operatorname{Pr}\!\left( S_n^{\mathrm{m}} \ge r n \right) \ge 1 - \varepsilon.
\end{equation}
Importantly, the dimension of the quantum memory depends only on \(r\) and is independent of \(\varepsilon\) and \(m\). 
This guarantees that the quantum memory can be reused to achieve the stated performance for an arbitrary number of iterations. 

\section{Conclusions}

In this work, we introduced a unified framework for local state discrimination that captures both catalytic assistance and protocols involving a reusable quantum memory, and formalizes their performance through achievable success rates across repeated rounds. This approach allows a direct comparison of strategies based on local operations and classical communication, catalytic assistance, and memory assistance within a common theoretical setting.

Our first main result is a robust ``no advantage'' theorem for separable encodings: when the hiding states are separable, neither borrowing an entangled catalyst nor employing a reusable quantum memory can improve the optimal success probability. This establishes that separable data hiding schemes are fundamentally resistant to both catalytic and memory based attacks. These findings complement previously known advantageous features of separable states in quantum data hiding~\cite{DiVincenzo_2002,PhysRevLett.89.097905}.

Our second main result reveals a contrasting behavior for entangled encodings, where access to a reusable quantum memory offers a significant advantage. We explicitly construct quantum states that are nearly indistinguishable under LOCC, yet become almost perfectly distinguishable when the communicating parties share a finite dimensional quantum memory that can be reused across successive rounds.

Taken together, these findings clarify the conditions under which additional quantum resources, entangled catalyst and reusable quantum memory, can overcome the limits of quantum data hiding. From a practical standpoint, they suggest a clear design principle: use separable encodings when robustness against advanced attacks, such as those equipped with catalytic resources or reusable quantum memories, is required. At the same time, one should recognize that entangled encodings may remain susceptible to discrimination strategies that exploit quantum memory.

Our results give rise to several open questions. 
While we have shown that certain data hiding states can be perfectly distinguished with the aid of a quantum memory, it remains unclear whether some data hiding schemes can also be broken using a quantum catalyst. 
A further challenge is to develop a general characterization of data hiding states that remain robust in the presence of catalytic and memory-assisted strategies. 
Beyond separable states, promising candidates include states with a positive partial transpose. 
Exploring these questions will clarify the fundamental boundary between global and local information access, lead to explicit criteria for catalytic and memory-assisted discrimination, and guide the design of quantum data-hiding schemes that are both robust and secure.

\section{Methods}

We now introduce the main definitions used throughout this work. We begin with the definition of separable states. A bipartite state $\rho^{AB}$ is considered separable if it can be expressed as a probabilistic mixture of product states~\cite{PhysRevA.40.4277}:
\begin{equation}
    \rho^{AB} = \sum_{x}p_x\, \psi^{A}_{x}\otimes\phi^{B}_{x},
\end{equation}
where $\{p_x\}_{x}$ is a probability distribution, and $\psi^{A}_{x}$ and $\phi^{B}_{x}$ are pure states. The set of all separable states is denoted by $\mathrm{SEP}$. Any state which is not separable is called entangled.

Any protocol based on local operations and classical communication (LOCC) acting on a bipartite quantum state \(\rho^{AB}\) can be written as~\cite{PhysRevLett.78.2275,PhysRevA.59.1070,Donald_2002,Chitambar2014}
\begin{equation}
\Lambda_{\mathrm{LOCC}}(\rho^{AB}) = \sum_{i} A_{i} \otimes B_{i} \,(\rho^{AB}) \, A_{i}^{\dagger} \otimes B_{i}^{\dagger},
\end{equation}
where \(A_i \otimes B_i\) are local Kraus operators associated with the LOCC protocol. 

A positive operator-valued measure (POVM) \(\{M_i^{AB}\}\) is said to be \emph{LOCC implementable} if each element $M_i^{AB}$ can be expressed as
\begin{equation}
M_i^{AB} = A_i^{\dagger} A_i \otimes B_i^{\dagger} B_i,
\end{equation}
for some local Kraus operators \(A_i\) and \(B_i\) arising from an LOCC protocol. 

An \emph{LOCC measurement channel} is a quantum-to-classical channel of the form~\cite{Matthews_2009}
\begin{equation}
\mathcal{M}(\rho^{AB}) = \sum_i \operatorname{Tr}\!\left[M_i^{AB} \rho^{AB}\right] \ket{ii}\!\bra{ii}^{AB},
\end{equation}
where \(\{M_i^{AB}\}\) is an LOCC implementable POVM.

The \emph{LOCC norm} of an operator \(X\) is defined as~\cite{Matthews_2009}
\begin{equation}
\Vert X \Vert_{\mathrm{LOCC}} = \sup_{\mathcal{M}} \Vert \mathcal{M}(X) \Vert_1,
\end{equation}
where the supremum is taken over all LOCC measurement channels \(\mathcal{M}\). 
By the data processing inequality for the trace norm, it follows that 
\begin{equation}
\Vert \rho_0^{AB} - \rho_1^{AB} \Vert_{\mathrm{LOCC}} \leq \Vert \rho_0^{AB} - \rho_1^{AB} \Vert_1,
\end{equation}
for any pair of quantum states \(\rho_0\) and \(\rho_1\).

\subsection*{Proof of Theorem~\ref{thm:Separable}}\label{sec:Prf_Catalyst}

In this section, we present the proof of Theorem~\ref{thm:Separable}. 
We begin by establishing the result for the catalytic setting in Proposition~\ref{prop:Catalytic}, and then extend the argument to the quantum memory setting in Proposition~\ref{prop:QuantumMemory}.

\begin{proposition} \label{prop:Catalytic}
	For separable states $\rho_{0,1} \in \mathrm{SEP}$, quantum catalysis cannot improve the optimal success probability:
    \begin{equation}
        R_{\mathrm{c}}(\rho_{0},\rho_{1})=P_{\mathrm{LOCC}}(\rho_{0},\rho_{1})
    \end{equation}
\end{proposition}

\begin{proof}
	Assume, toward a contradiction, that there exists a catalyst state $\tau^{A'B'}$ such that 
    \begin{equation}
        R_{\mathrm{c}}(\rho_{0},\rho_{1})>P_{\mathrm{LOCC}}(\rho_{0},\rho_{1}).
    \end{equation}
    Without loss of generality, we can assume that $\tau^{A'B'}$ is entangled because the addition of a separable catalyst is achievable by LOCC. 
    From the setting considered and the definition of $R_{\mathrm{c}}(\rho_{0},\rho_{1})$, there exists also an LOCC protocol \(\Lambda_{\mathrm{LOCC}}\) such that Eqs.~(\ref{eq:Catalyst-1}) and (\ref{eq:Catalyst-2}) are fulfilled. Let now $\mathcal M_{AB}$ be an LOCC discrimination channel which is optimal for discriminating the states $\sigma_0^{AB}$ and $\sigma_1^{AB}$, which are obtainable from $\rho_0^{AB}$ and $\rho_1^{AB}$ via catalytic LOCC. The overall LOCC protocol can then be written as $\mathcal{M}_{\tau}=\mathcal{M}_{AB}\circ\Lambda_{\mathrm{LOCC}}$. With these definitions we have
	\begin{equation}
		\Vert\mathcal{M}_{\tau}(\rho^{AB}_{0}\otimes\tau^{A'B'}-\rho^{AB}_{1}\otimes\tau^{A'B'})\Vert_{1}> \Vert\rho^{AB}_{0}-\rho^{AB}_{1}\Vert_{\mathrm{LOCC}}.
	\end{equation}
    The achievable success probability to distinguish the states $\rho^{AB}_{0}$ and $\rho^{AB}_{1}$ with this procedure can then be written as 
\begin{equation}
p_{\tau}=\frac{1}{2}+\frac{1}{4}\Vert\mathcal{M}_{\tau}(\rho_{0}^{AB}\otimes\tau^{A'B'}-\rho_{1}^{AB}\otimes\tau^{A'B'})\Vert_{1}.
\end{equation}

As we will now show, the existence of such a protocol could be exploited by Alice and Bob for learning if the catalyst is in a separable or in an entangled state. In particular, assume now that with probability $1/2$ the catalyst is initialized in the correct state $\tau^{A'B'}$, and with the same probability it is initialized in a separable state $\gamma^{A'B'}\in\operatorname{SEP}$ which is not orthogonal to $\tau^{A'B'}$. In more detail, let the initial state of the catalyst be denoted by $\eta_1^{A'B'}$, and choose parameter $\delta$ in the range 
\begin{equation}
0<\delta<\frac{p_{\tau}-p_\mathrm{LOCC}}{2}, \label{eq:DeltaRange}
\end{equation}
where we defined 
\begin{equation}
    p_\mathrm{LOCC} = P_{\mathrm{LOCC}}(\rho_{0},\rho_{1}).
\end{equation}

Alice and Bob now repeat the following steps $n$ times to obtain the classical random variable $X_{j}$, which will then be used to distinguish $\gamma^{A'B'}$ and $\tau^{A'B'}$. In the following, $\eta_j^{A'B'}$ denotes the state of the system $A'B'$ during $j$-th iteration.
	\begin{enumerate}
			\item Alice and Bob choose $Z_{j}\in\{0,1\}$ uniformly at random.
			\item They set $X_{j}=\perp$.
			\item If $Z_j=0$, they prepare the system registers in the state $\rho^{AB}_{0}$ using LOCC, else they prepare the state $\rho^{AB}_{1}$.
			\item They apply the LOCC protocol $\mathcal{M}_{\tau}$ on $\rho^{AB}_{Z_{j}}\otimes\eta_{j}^{A'B'}$.
			\item They measure the system registers in the computational basis. 
			\item They obtain the result $Y_{j}$. If $Y_{j}=Z_{j}$, then they set $X_{j}=1$, else $X_{j}=0$. 
			\item They update $j$ to $j+1$.
	\end{enumerate}
	After $n$ rounds, Alice and Bob compute $S_{n}=\sum_{i=1}^{n}X_{j}$. If 
    \begin{equation}
    \left|\frac{S_{n}}{n}-p_{\tau}\right|\leq\delta, \label{eq:CorrectGuessCondition}
    \end{equation}
then Alice and Bob guess that the initial state of the catalyst register was the entangled state $\tau^{A'B'}$, else they guess that the initial state of the catalyst register was the separable state $\gamma^{A'B'}$.

In the following, we will prove that this protocol can achieve perfect discrimination of $\tau^{A'B'}$ and $\gamma^{A'B'}$, leading us to the desired contradiction. For this, we will consider two cases, namely $\eta_1^{A'B'} = \tau^{A'B'}$ (Case 1) and $\eta_1^{A'B'} = \gamma^{A'B'}$ (Case 2).

\textbf{Case 1:} If the initial state is $\eta_{1}^{A'B'}=\tau^{A'B'}\notin\operatorname{SEP}$, we know, by assumption, that the entangled catalyst is recovered perfectly and $\eta_{j}^{A'B'}=\tau^{A'B'}$ for all $j\leq n$. Moreover, each round of the process will be independent and identically distributed, which means that $X_j$ is an i.i.d. random variable in this case. Using Hoeffding's inequality~\cite{H63}, we get the following inequality for all $\delta > 0$ and all $n$:
\begin{align}
			\operatorname{Pr}\left(\left\vert \frac{S_{n}}{n} - p_{\tau}\right\vert\leq \delta\right)
			&=\operatorname{Pr}\left(\vert S_{n} - np_{\tau}\vert\leq n\delta\right)\nonumber\\
			&\geq 1-2\exp(-2n\delta^2),
\end{align}

Hence, recalling Eq.~(\ref{eq:CorrectGuessCondition}), the probability that Alice and Bob correctly guess the initial state of the catalyst in this setting is bounded as follows:
\begin{equation}
P_{\mathrm{corr}} (\tau)\geq1-2\exp(-2n\delta^{2}).
\end{equation}
This completes the analysis for Case 1.

\textbf{Case 2:} If the initial state of the catalyst is $\eta_{1}^{A'B'}=\gamma^{A'B'}\in\operatorname{SEP}$, then we cannot assume that each round of the process will be independent and identically distributed. At the end of each round, the state of the $A'B'$ register may change. In the first round, the total state prior to the measurement can be written as
		\begin{equation}
			\mathcal{M}_{\tau}(\rho^{AB}_{Z_{1}}\otimes\gamma^{A'B'})=\omega^{ABA'B'}_{Z_1}.
		\end{equation}
		Depending on the outcome of the measurement on the register $AB$, the system $A'B'$ is in the state $\eta_{2}^{A'B'}=\omega^{A'B'}_{X_1Z_1}$. In the second round, the total state prior to the measurement takes the form 
		\begin{equation}
			\mathcal{M}_{\tau}(\rho^{AB}_{Z_{2}}\otimes\omega^{A'B'}_{X_1Z_1})=\omega^{ABA'B'}_{Z_2X_1Z_1}.
		\end{equation}
		Depending on the outcome of the measurement on the system $AB$, the system $A'B'$ is in the state $\eta_{3}^{A'B'}=\omega^{A'B'}_{X_2Z_2X_1Z_1}$. After $j$ rounds, the state of the register $A'B'$ is $\eta_{j}^{A'B'}= \omega^{A'B'}_{X_{j}Z_{j}\ldots X_{1}Z_{1}}$. 

		Note that within each round, the state of the registers $AB$ is separable, and the channel applied is LOCC. Since the register $A'B'$ was initially in a separable state $\gamma^{A'B'}$, it will remain in a separable state throughout the protocol. Note that the addition of a separable state cannot increase the probability of distinguishing between $\rho_{0}^{AB}$ and $\rho_{1}^{AB}$. Hence, for the $j+1$-th round, the probability for Alice and Bob to make a correct guess can be bounded as 
		\begin{equation}
			\operatorname{Pr}(X_{j+1}=1|X_{1}\ldots X_{j}) \leq \frac{1}{2} + \frac{1}{4}\Vert\rho^{AB}_{0}-\rho^{AB}_{1}\Vert_{\mathrm{LOCC}} = p_\mathrm{LOCC}.
		\end{equation}
		
        Let us now consider the random variable  
        $C_j\coloneqq S_j-jp_{\mathrm{LOCC}}$ for $1\leq j\leq n$ and $C_0=0$. For the expected value $\boldsymbol{\mathrm E}$ we obtain
		\begin{align}
			&\boldsymbol{\mathrm E}[C_{j+1}|C_j,\ldots,C_1]\nonumber\\
			&= \sum_{X_{j+1}=0}^{1}\left(\sum_{i=1}^{j}X_{i}+X_{j+1}-(j+1)p_\mathrm{LOCC}\right)\operatorname{Pr}(X_{j+1}|X_{1}\ldots X_{j})\nonumber\\
			&= \sum_{i=1}^{j}X_{i}-(j+1)p_\mathrm{LOCC}+\sum_{X_{j+1}=0}^{1}\left(X_{j+1}\right)\operatorname{Pr}(X_{j+1}|X_{1}\ldots X_{j})\nonumber\\
			&\leq C_j -p_\mathrm{LOCC}+ p_\mathrm{LOCC}=C_j.
		\end{align}
		Hence, $C_j$ is a supermartingale. Additionally, it is clear that
		\begin{equation}
			C_j - C_{j-1}\leq 1.
		\end{equation}
		Then, using Azuma's inequality~\cite{azuma1967weighted} for supermartingales, we get the following inequality for all $\delta > 0$ and all $n$: 
		\begin{align} \label{eq:Azuma}
			\operatorname{Pr}\left(\frac{S_{n}}{n} - p_\mathrm{LOCC}\geq \delta\right)
			&=\operatorname{Pr}(S_n-np_\mathrm{LOCC}\geq n\delta)\nonumber\\
			&\leq \exp\left(\frac{-(n\delta)^2}{2n}\right)=\exp\left(\frac{-n\delta^2}{2}\right).
		\end{align}

Moreover, note the following inequality:
\begin{align}
\operatorname{Pr}\left(\left|\frac{S_{n}}{n}-p_{\tau}\right|\geq\delta\right) & =\operatorname{Pr}\left(\frac{S_{n}}{n}-p_{\tau}\geq\delta\right)+\operatorname{Pr}\left(\frac{S_{n}}{n}-p_{\tau}\leq-\delta\right)\nonumber \\
 & \geq\operatorname{Pr}\left(\frac{S_{n}}{n}-p_{\tau}\leq-\delta\right).
\end{align}
Recalling that $\delta$ fulfills $0 < \delta < (p_\tau - p_\mathrm{LOCC})/2$, it immediately follows that $p_\mathrm{LOCC} + \delta < p_\tau - \delta$. We thus have 
\begin{align}
\operatorname{Pr}\left(\frac{S_{n}}{n}-p_{\tau}\geq-\delta\right) & =\operatorname{Pr}\left(\frac{S_{n}}{n}\geq p_{\tau}-\delta\right)\\
 & \leq\operatorname{Pr}\left(\frac{S_{n}}{n}\geq p_\mathrm{LOCC}+\delta\right)\nonumber \\
 & =\operatorname{Pr}\left(\frac{S_{n}}{n}-p_\mathrm{LOCC}\geq\delta\right) \nonumber \\
 & \leq\exp\left(\frac{-n\delta^{2}}{2}\right).\nonumber 
\end{align}
Hence, the probability that Alice and Bob make a correct guess in this setting is bounded as 
\begin{align}
P_{\mathrm{corr}}(\gamma) & =\operatorname{Pr}\left(\left|\frac{S_{n}}{n}-p_{\tau}\right|>\delta\right)\\
 & \geq\operatorname{Pr}\left(\frac{S_{n}}{n}-p_{\tau}<-\delta\right)\nonumber \\
 & =1-\operatorname{Pr}\left(\frac{S_{n}}{n}-p_{\tau}\geq-\delta\right)\nonumber \\
 & \geq1-\exp\left(\frac{-n\delta^{2}}{2}\right).\nonumber 
\end{align}
This concludes the analysis of Case 2.

From the analysis above, we get that the proposed LOCC protocol succeeds in distinguishing $\gamma^{A'B'}$ and $\tau^{A'B'}$ with the overall probability 
\begin{align}
P_{\mathrm{corr}} & =\frac{1}{2}\left[P_{\mathrm{corr}}(\tau)+P_{\mathrm{corr}}(\gamma)\right]\\
 & \geq\frac{1}{2}\left[1-2\exp\left(-2n\delta^{2}\right)\right]+\frac{1}{2}\left[1-\exp\left(-n\delta^{2}/2\right)\right].\nonumber 
\end{align}
Moreover, we can choose arbitrary integer $n$ and arbitrary $\delta$ in the range given in Eq.~(\ref{eq:DeltaRange}).  
    
By assumption, it holds that $\Vert\tau^{A'B'}-\gamma^{A'B'}\Vert_{1}\neq 2$, which means that Alice and Bob can achieve
\begin{equation}
P_{\mathrm{corr}}>P_{\mathrm{opt}}\left(\tau^{A'B'},\gamma^{A'B'}\right)=\frac{1}{2}+\frac{\Vert\tau^{A'B'}-\gamma^{A'B'}\Vert_{1}}{4},
\end{equation}
whenever $n$ fulfills
\begin{align}
n> & \max\left\{ \frac{1}{2\delta^{2}}\left[-\ln\left(\frac{1}{4}-\frac{\Vert\tau^{A'B'}-\gamma^{A'B'}\Vert_{1}}{8}\right)\right],\right.\\
 & \left.\frac{2}{\delta^{2}}\left[-\ln\left(\frac{1}{2}-\frac{\Vert\tau^{A'B'}-\gamma^{A'B'}\Vert_{1}}{4}\right)\right]\right\} .\nonumber 
\end{align}

	Hence, it would appear that Alice and Bob can distinguish between two non-orthogonal states $\gamma^{A'B'}$ and $\tau^{A'B'}$ using the above-mentioned LOCC protocol with probability greater than the maximum of $P_\mathrm{opt}$. This is a contradiction. Hence, proved. 
\end{proof}

To complete the proof of Theorem~\ref{thm:Separable}, we will now adjust the methods presented above, making them applicable to the setting with a reusable quantum memory.

\begin{proposition} \label{prop:QuantumMemory}
	For separable states $\rho_{0,1} \in \mathrm{SEP}$, quantum memory can not improve the optimal success probability:
    \begin{equation}
    R_{\mathrm{m}}(\rho_{0},\rho_{1})=P_{\mathrm{LOCC}}(\rho_{0},\rho_{1})
    \end{equation}
\end{proposition}
\begin{proof}
	Our proof will proceed by contradiction, analogously to the proof of Proposition~\ref{prop:Catalytic}. We thus assume that there is an advantage provided by a quantum memory in distinguishing $\rho_{0}$ and $\rho_{1}$, i.e., 
    \begin{equation}
    R_{\mathrm{m}}(\rho_{0},\rho_{1})>P_{\mathrm{LOCC}}(\rho_{0},\rho_{1}).
    \end{equation}
    Let further $\mu^{A'B'}$ be an initial state of the quantum memory, and $\{\mathcal{M}_i\}$ a sequence of LOCC protocols achieving the rate $R_{\mathrm{m}}(\rho_{0},\rho_{1})$. We further define $r_\mathrm{m} = R_{\mathrm{m}}(\rho_{0},\rho_{1})$, and $\delta > 0$ such that $r_\mathrm{m}>p_\mathrm{LOCC}+\delta$.

    Analogously to the proof of Proposition~\ref{prop:Catalytic}, we will now show that such an advantage could be exploited by Alice and Bob to perfectly distinguish if the initial state of the quantum memory was $\mu^{A'B'}$, or whether it was a separable state $\gamma^{A'B'}\in\operatorname{SEP}$, where $\gamma^{A'B'}$ is non-orthogonal to $\mu^{A'B'}$.

	Assume now that the quantum memory is initialized either in the state $\mu^{A'B'}$, or in the state $\gamma^{A'B'}$, each with probability $1/2$. We will denote the initial state with $\eta_{1}^{A'B'}\in\{\gamma^{A'B'},\mu^{A'B'}\}$. Alice and Bob, then, repeat the following steps $n$ times to obtain the classical random variable $X_{j}$ which will be used to distinguish $\gamma^{A'B'}$ and $\mu^{A'B'}$. Also in this setting, $\eta_j^{A'B'}$ denotes the state of the system $A'B'$ during $j$-th iteration.
	\begin{enumerate}
			\item Alice and Bob choose $Z_{j}\in\{0,1\}$ uniformly at random.
			\item They set $X_{j}=\perp$.
			\item If $Z_j=0$, they prepare the system registers in the state $\rho^{AB}_{0}$ using LOCC, else they prepare the state $\rho^{AB}_{1}$.
			\item They apply the LOCC protocol $\mathcal{M}_{j}$ on $\rho^{AB}_{Z_{j}}\otimes\eta_{j}^{A'B'}$. 
			\item They obtain the result $Y_{j}$. If $Y_{j}=Z_{j}$, they set $X_{j}=1$, else $X_{j}=0$. 
			\item They update $j$ to $j+1$.
	\end{enumerate}
	In the above procedure, each of the LOCC protocols $\mathcal{M}_{j}$ can depend on the outcomes of the previous rounds since we are considering a procedure involving a quantum memory.
    
    After $n$ rounds, Alice and Bob compute $S_{n}=\sum_{i=1}^{n}X_{j}$. If 
\begin{equation}
    \frac{S_n}{n}-p_\mathrm{LOCC}\geq \delta,
\end{equation}
then Alice and Bob guess that the initial state of the memory register was the entangled state $\mu^{A'B'}$, else they guess that initial state of the memory register was the separable state $\gamma^{A'B'}$. 

    We will now show that this procedure can be used to perfectly detect whether the quantum memory was initially in the state $\mu^{A'B'}$ or in a separable state $\gamma^{A'B'}$, which will lead to the desired contradiction. Analogously to the catalytic setting, we will consider two cases, namely $\eta_{1}^{A'B'}=\mu^{A'B'}$ (Case 1) and $\eta_{1}^{A'B'}=\gamma^{A'B'}$ (Case 2).
    
	\textbf{Case 1:} If the initial state is $\eta_{1}^{A'B'}=\mu^{A'B'}$, by assumption, for every \(\varepsilon > 0\) and \(m > 0\), there exists some \(n \ge m\)  such that 
		\begin{equation}
			\operatorname{Pr}\!\left( S_n \ge nr_\mathrm{m} \right) \ge 1 - \varepsilon.
		\end{equation}
		Recalling that $r_\mathrm{m}> p_\mathrm{LOCC}+\delta$, we get
		\begin{equation}
			\operatorname{Pr}\!\left( S_n - np_\mathrm{LOCC} \ge n\delta  \right) \ge 1 - \varepsilon.
		\end{equation}
		Hence, the probability that Alice and Bob make a correct guess in this setting is bounded as 
        \begin{equation}
        P_{\mathrm{corr}}(\mu)\geq1-\varepsilon.
        \end{equation}
    This concludes the analysis of Case 1.
		
    \textbf{Case 2:} If the initial state is $\eta_{1}^{A'B'}=\gamma^{A'B'}\in \mathrm{SEP}$, the analysis follows the same lines of reasoning as Case 2 for the catalytic setting. Also in this setting, we arrive at the inequality~(\ref{eq:Azuma}). Hence, the probability that Alice and Bob make a correct guess in this setting is bounded as 
    \begin{equation}
P_{\mathrm{corr}}(\gamma)=\mathrm{Pr}\left(\frac{S_{n}}{n}-p_{\mathrm{LOCC}}<\delta\right)\geq1-\exp\left(\frac{-n\delta^{2}}{2}\right).
\end{equation}
	This concludes the analysis of Case 2.
    
	From the analysis above, we get that the protocol succeeds in distinguishing $\gamma^{A'B'}$ and $\mu^{A'B'}$ with probability 
    \begin{align}
        P_{\mathrm{corr}} & =\frac{1}{2}\left[P_{\mathrm{corr}}(\mu)+P_{\mathrm{corr}}(\gamma)\right]\\
        & \geq\frac{1}{2}\left(1-\varepsilon\right)+\frac{1}{2}\left(1-\exp\left(-n\delta^{2}/2\right)\right). \nonumber
    \end{align}
	Recall that we can choose an arbitrarily small $\varepsilon > 0$ and an arbitrarily large $n$. Since we assumed that $\mu^{A'B'}$ and $\gamma^{A'B'}$ are nonorthogonal, we see that for sufficiently large $n$ and sufficiently small $\varepsilon > 0$ we have
    \begin{equation}
    P_{\mathrm{corr}}>P_{\mathrm{opt}}(\mu,\gamma).
    \end{equation}
    
	Hence, it would appear that Alice and Bob can distinguish two non-orthogonal states $\gamma^{A'B'}$ and $\mu^{A'B'}$ using the above-mentioned protocol with probability greater than the maximum of $P_\mathrm{opt}$. 
	This is the desired contradiction, and the proof is complete.
\end{proof}

\subsection*{Proof of Theorem \ref{thm:2}}
We will now provide a construction for two states $\rho_0^{AB}$ and $\rho_1^{AB}$ fulfilling Eqs.~(\ref{eq:Theorem2-1}) and (\ref{eq:Theorem2-2}). 

In this construction, each local system consists of two subsystems, that is, \(A = A_1A_2\) and \(B = B_1B_2\). 
For some $\varepsilon>0$ consider two states \(\sigma_0^{A_1B_1}\) and \(\sigma_1^{A_1B_1}\) that satisfy
\begin{align}
P_{\mathrm{opt}}(\sigma_{0},\sigma_{1}) & =1,\\
P_{\mathrm{LOCC}}(\sigma_{0},\sigma_{1}) & \le\frac{1}{2}+\varepsilon. \label{eq:Sigma01}
\end{align}
It is known that such states exist for any $\varepsilon > 0$~\cite{PhysRevLett.89.097905,Aubrun_2015,mele2025optimisingquantumdatahiding}. 

For some $\varepsilon' > 0$, let \(\ket{\psi}^{A_2B_2}\) be an entangled state satisfying
\begin{align}
\left\Vert \ket{\psi}\!\bra{\psi}^{A_{2}B_{2}} - \ket{00}\!\bra{00}^{A_{2}B_{2}} \right\Vert_1 &< \varepsilon', \label{eq:AlmostProduct}\\
S(\psi^{A_2}) &> \log_2 d_{A_1} \label{eq:EntanglementEntropy}
\end{align}
with von Neumann entropy $S(\rho)=-\mathrm{Tr}(\rho\log_{2}\rho)$. An example for a state with these properties can be given as 
\begin{equation}
\ket{\psi}=\sqrt{\lambda}\ket{00}+\sqrt{\frac{1-\lambda}{d_{A_{2}}-1}}\sum_{i=1}^{d_{A_{2}}-1}\ket{ii}
\end{equation}
with $\lambda \in (0,1)$. Noting that $|\!\braket{00|\psi}\!|^{2}=\lambda$ and using the inequality 
\begin{equation}
\left\Vert \rho-\sigma\right\Vert _{1}\leq2\sqrt{1-F(\rho,\sigma)}
\end{equation}
with fidelity $F(\rho,\sigma)=\left(\mathrm{Tr}\sqrt{\sqrt{\rho}\sigma\sqrt{\rho}}\right)^{2}$ it immediately follows that 
\begin{equation}
\left\Vert \ket{\psi}\!\bra{\psi}-\ket{00}\!\bra{00}\right\Vert _{1}\leq2\sqrt{1-\lambda}.
\end{equation}
It follows that Eq.~(\ref{eq:AlmostProduct}) is fulfilled whenever $\lambda$ fulfills 
\begin{equation}
\lambda>1-\frac{(\varepsilon')^{2}}{4}.
\end{equation}
We further have 
\begin{equation}
S(\psi^{A_{2}})=-\lambda\log_{2}\lambda-(1-\lambda)\log_{2}\frac{1-\lambda}{d_{A_{2}}-1}.
\end{equation}
It is clear that for any value of $\lambda \in (0,1)$ we can fulfill Eq.~(\ref{eq:EntanglementEntropy}) by choosing large enough $d_{A_2}$.

With these ingredients, we define the states
\begin{align}
\rho_{0}^{AB} &= \sigma_{0}^{A_{1}B_{1}} \otimes \ket{\psi}\!\bra{\psi}^{A_{2}B_{2}},  \\
\rho_{1}^{AB} &= \sigma_{1}^{A_{1}B_{1}} \otimes \ket{\psi}\!\bra{\psi}^{A_{2}B_{2}}.
\end{align}
As we will see in the following, Eq.~(\ref{eq:AlmostProduct}) implies that the states \(\rho_0\) and \(\rho_1\) satisfy
\begin{equation} \label{eq:PloccBound}
P_{\mathrm{LOCC}}(\rho_{0}, \rho_{1}) \le \varepsilon + \frac{1+\varepsilon'}{2},
\end{equation}
demonstrating that they remain almost indistinguishable under LOCC.

To prove this, note that Eq.~(\ref{eq:Sigma01}) is equivalent to 
\begin{equation}
\bigl\Vert\sigma_{0}-\sigma_{1}\bigr\Vert_{\mathrm{LOCC}}\leq4\varepsilon.\label{eq:Sigma12}
\end{equation}
It is clear that $\rho_{0}^{AB}$
and $\rho_{1}^{AB}$ are orthogonal whenever this is true for $\sigma_{0}$
and $\sigma_{1}$. We will now analyze the LOCC norm of $\rho_{0}-\rho_{1}$.
In the following, $\mathcal{M}$ denotes an optimal LOCC discrimination
protocol for the states $\rho_{0}$ and $\rho_{1}$, i.e., 
\begin{align}
\left\Vert \rho_{0}-\rho_{1}\right\Vert _{\mathrm{LOCC}} & =\left\Vert \mathcal{M}\left[\rho_{0}-\rho_{1}\right]\right\Vert _{1}\label{eq:OptimalLOCC}\\
 & =\left\Vert \mathcal{M}\left[\sigma_{0}^{A_{1}B_{1}}\otimes\psi^{A_{2}B_{2}}\right]-\mathcal{M}\left[\sigma_{1}^{A_{1}B_{1}}\otimes\psi^{A_{2}B_{2}}\right]\right\Vert _{1}.\nonumber 
\end{align}
Using triangle inequality for the trace norm we further find 
\begin{align}
 & \left\Vert \mathcal{M}\left[\sigma_{0}^{A_{1}B_{1}}\otimes\psi^{A_{2}B_{2}}\right]-\mathcal{M}\left[\sigma_{1}^{A_{1}B_{1}}\otimes\psi^{A_{2}B_{2}}\right]\right\Vert _{1}\label{eq:Triangle-1}\\
 & \leq\left\Vert \mathcal{M}\left[\sigma_{0}^{A_{1}B_{1}}\otimes\psi^{A_{2}B_{2}}\right]-\mathcal{M}\left[\sigma_{1}^{A_{1}B_{1}}\otimes\ket{00}\!\bra{00}^{A_{2}B_{2}}\right]\right\Vert _{1}\nonumber \\
 & +\left\Vert \mathcal{M}\left[\sigma_{1}^{A_{1}B_{1}}\otimes\ket{00}\!\bra{00}^{A_{2}B_{2}}\right]-\mathcal{M}\left[\sigma_{1}^{A_{1}B_{1}}\otimes\psi^{A_{2}B_{2}}\right]\right\Vert _{1},\nonumber 
\end{align}
and similarly 
\begin{align}
 & \left\Vert \mathcal{M}\left[\sigma_{0}^{A_{1}B_{1}}\otimes\psi^{A_{2}B_{2}}\right]-\mathcal{M}\left[\sigma_{1}^{A_{1}B_{1}}\otimes\ket{00}\!\bra{00}^{A_{2}B_{2}}\right]\right\Vert _{1}\label{eq:Triangle-2}\\
 & \leq\left\Vert \mathcal{M}\left[\sigma_{0}^{A_{1}B_{1}}\otimes\ket{00}\!\bra{00}^{A_{2}B_{2}}\right]-\mathcal{M}\left[\sigma_{1}^{A_{1}B_{1}}\otimes\ket{00}\!\bra{00}^{A_{2}B_{2}}\right]\right\Vert _{1}\nonumber \\
 & +\left\Vert \mathcal{M}\left[\sigma_{0}^{A_{1}B_{1}}\otimes\ket{00}\!\bra{00}^{A_{2}B_{2}}\right]-\mathcal{M}\left[\sigma_{0}^{A_{1}B_{1}}\otimes\psi^{A_{2}B_{2}}\right]\right\Vert _{1}.\nonumber 
\end{align}
Using Eq.~(\ref{eq:Triangle-2}) in Eq.~(\ref{eq:Triangle-1}) we
find 
\begin{align}
 & \left\Vert \mathcal{M}\left[\sigma_{0}^{A_{1}B_{1}}\otimes\psi^{A_{2}B_{2}}\right]-\mathcal{M}\left[\sigma_{1}^{A_{1}B_{1}}\otimes\psi^{A_{2}B_{2}}\right]\right\Vert _{1}\\
 & \leq\left\Vert \mathcal{M}\left[\sigma_{0}^{A_{1}B_{1}}\otimes\ket{00}\!\bra{00}^{A_{2}B_{2}}\right]-\mathcal{M}\left[\sigma_{1}^{A_{1}B_{1}}\otimes\ket{00}\!\bra{00}^{A_{2}B_{2}}\right]\right\Vert _{1}\nonumber \\
 & +\left\Vert \mathcal{M}\left[\sigma_{0}^{A_{1}B_{1}}\otimes\ket{00}\!\bra{00}^{A_{2}B_{2}}\right]-\mathcal{M}\left[\sigma_{0}^{A_{1}B_{1}}\otimes\psi^{A_{2}B_{2}}\right]\right\Vert _{1}\nonumber \\
 & +\left\Vert \mathcal{M}\left[\sigma_{1}^{A_{1}B_{1}}\otimes\ket{00}\!\bra{00}^{A_{2}B_{2}}\right]-\mathcal{M}\left[\sigma_{1}^{A_{1}B_{1}}\otimes\psi^{A_{2}B_{2}}\right]\right\Vert _{1}.\nonumber 
\end{align}
Using this in Eq.~(\ref{eq:OptimalLOCC}) we further obtain
\begin{align}
 & \left\Vert \rho_{0}-\rho_{1}\right\Vert _{\mathrm{LOCC}}\label{eq:LOCCnormBound-2}\\
 & \leq\left\Vert \mathcal{M}\left[\sigma_{0}^{A_{1}B_{1}}\otimes\ket{00}\!\bra{00}^{A_{2}B_{2}}\right]-\mathcal{M}\left[\sigma_{1}^{A_{1}B_{1}}\otimes\ket{00}\!\bra{00}^{A_{2}B_{2}}\right]\right\Vert _{1}\nonumber \\
 & +\left\Vert \mathcal{M}\left[\sigma_{0}^{A_{1}B_{1}}\otimes\ket{00}\!\bra{00}^{A_{2}B_{2}}\right]-\mathcal{M}\left[\sigma_{0}^{A_{1}B_{1}}\otimes\psi^{A_{2}B_{2}}\right]\right\Vert _{1}\nonumber \\
 & +\left\Vert \mathcal{M}\left[\sigma_{1}^{A_{1}B_{1}}\otimes\ket{00}\!\bra{00}^{A_{2}B_{2}}\right]-\mathcal{M}\left[\sigma_{1}^{A_{1}B_{1}}\otimes\psi^{A_{2}B_{2}}\right]\right\Vert _{1}.\nonumber 
\end{align}

In the next step, consider the map 
\begin{equation}
\widetilde{\mathcal{M}}\left[\mu^{A_{1}B_{1}}\right]=\mathcal{M}\left[\mu^{A_{1}B_{1}}\otimes\ket{00}\!\bra{00}^{A_{2}B_{2}}\right].
\end{equation}
Since the attachment of a product state can be implemented via LOCC, it is straightforward to see that
\begin{align}
\left\Vert \widetilde{\mathcal{M}}\left[\sigma_{0}^{A_{1}B_{1}}-\sigma_{1}^{A_{1}B_{1}}\right]\right\Vert _{1} & \leq\left\Vert \sigma_{0}^{A_{1}B_{1}}-\sigma_{1}^{A_{1}B_{1}}\right\Vert _{\mathrm{LOCC}}.
\end{align}
Moreover, using the data processing inequality for trace norm we see
that 
\begin{align}
 & \left\Vert \mathcal{M}\left[\gamma^{A_{1}B_{1}}\otimes\ket{00}\!\bra{00}^{A_{2}B_{2}}\right]-\mathcal{M}\left[\gamma^{A_{1}B_{1}}\otimes\psi^{A_{2}B_{2}}\right]\right\Vert _{1}\nonumber \\
 & \leq\left\Vert \ket{00}\!\bra{00}^{A_{2}B_{2}}-\psi^{A_{2}B_{2}}\right\Vert _{1}
\end{align}
for any state $\gamma^{AB}$. Using these results in Eq.~(\ref{eq:LOCCnormBound-2})
we find 
\begin{align}
\left\Vert \rho_{0}-\rho_{1}\right\Vert _{\mathrm{LOCC}} & \leq\left\Vert \sigma_{0}^{A_{1}B_{1}}-\sigma_{1}^{A_{1}B_{1}}\right\Vert _{\mathrm{LOCC}}\\
 & +2\left\Vert \ket{00}\!\bra{00}^{A_{2}B_{2}}-\psi^{A_{2}B_{2}}\right\Vert _{1}.\nonumber 
\end{align}
In the final step we use Eqs.~(\ref{eq:AlmostProduct}) and~(\ref{eq:Sigma12}), leading to 
\begin{align}
\left\Vert \rho_{0}-\rho_{1}\right\Vert _{\mathrm{LOCC}} & \leq4\varepsilon+2\varepsilon'.
\end{align}
Using this bound in Eq.~(\ref{eq:Plocc}) we arrive at the claimed inequality~(\ref{eq:PloccBound}). In particular, by choosing small enough $\varepsilon$ and $\varepsilon'$
we can achieve $P_{\mathrm{LOCC}}(\rho_{0},\rho_{1})<1/2+\varepsilon''$ for any $\varepsilon'' > 0$.

We will now show how Alice and Bob can distinguish $\rho_{0}$ and
$\rho_{1}$ by using a quantum memory. In the following, let $P_{n}$
be the probability to convert $\ket{\psi}^{\otimes n}$ into the state
\begin{equation}
\ket{\phi_{L_{n}}}=\frac{1}{\sqrt{L_{n}}}\sum_{i=0}^{L_{n}-1}\ket{ii}.
\end{equation}
Recall that for any $\varepsilon,\delta>0$ there exists some $r>0$
such that the following inequalities hold for all large enough $n$ \cite{Hayashi_2003}:
\begin{align}
P_{n} & \geq1-2^{-n(r-\delta)},\\
\log_{2}L_{n} & \geq n\left[S(\psi^{A_{2}})-\varepsilon\right].
\end{align}
Recalling that $S(\psi^{A_{2}})> \log_2 d_{A_{1}}$, it follows that for any
$\tilde{\varepsilon}$ there exists some $n$ such that the conversion
$\ket{\psi}^{\otimes n}\rightarrow\ket{\phi_{d_{A_{1}}}}^{\otimes n}$
is possible with probability
\begin{equation}
P_{n}\geq1-\tilde{\varepsilon}.\label{eq:Pn}
\end{equation}
In the following, we assume that $n$ and $\tilde{\varepsilon}$ are
chosen such that Eq.~(\ref{eq:Pn}) holds.

Assume now that Alice and Bob have access to a quantum memory $A'B'$
of dimension $d_{A'}=d_{B'}=d_{A_{2}}^{n}$. The initial state of
the quantum memory will be $\ket{\phi_{d_{A_{1}}^{n}}}$, which is
locally equivalent to $n$ copies of the state $\ket{\phi_{d_{A_{1}}}}$.
Note that Alice and Bob can use each of the copies of $\ket{\phi_{d_{A_{1}}}}$
to teleport Alice's part of $\sigma_{i}^{A_{1}B_{1}}$ to Bob, who
can then locally perfectly distinguish the states $\sigma_{0}$ and
$\sigma_{1}$. In each round Alice and Bob consume one copy of $\ket{\phi_{d_{A_{1}}}}$
and store one copy of $\ket{\psi}$ in the quantum memory. After $m$
rounds (with $m<n$) the quantum memory is in the state $\ket{\phi_{d_{A_{1}}}}^{\otimes n-m}\otimes\ket{\psi}^{\otimes m}$.
After $n$ rounds all copies of $\ket{\phi_{d_{A_{1}}}}$ have been
consumed, and the quantum memory is in the state $\ket{\psi}^{\otimes n}$.
As discussed above, by using LOCC Alice and Bob can convert $\ket{\psi}^{\otimes n}$
into $\ket{\phi_{d_{A_{1}}}}^{\otimes n}$ with probability $P_{n}\geq1-\tilde{\varepsilon}$.
If the conversion procedure fails, Alice and Bob will need another
$n$ rounds to establish a new instance of $\ket{\psi}^{\otimes n}$.
In this period they can perform standard LOCC state discrimination, which gives a success
rate of at least $1/2$. 

Choose some integer $k$ and assume now that Alice and Bob perform
$kn$ rounds of the local state discrimination procedure. We can think of
the procedure consisting of $k$ independent blocks, each having length
$n$. Within each of the blocks, Alice and Bob can perform the same
discrimination procedure independently, as described above. From the
above discussion it follows that the probability that Alice and Bob
can perfectly distinguish all states within each of the block is at
least $1-\tilde{\varepsilon}$. Let us now define 
\begin{equation}
r=1-\tilde{\varepsilon}-\delta
\end{equation}
with some $\delta>0$. Since each of the blocks is treated independently,
it is clear from Hoeffding's inequality~\cite{H63} that for any $\tilde{\varepsilon},\delta>0$ we can choose
some $n$ such that the following inequality holds for any $\varepsilon>0$
and all $k$ large enough:
\begin{equation}
P(S_{kn}\geq rkn)>1-\varepsilon.
\end{equation}
This completes the proof.

\section*{Acknowledgements}
We thank Saronath Halder for discussion. This work was supported by the National Science Centre Poland (Grant No. 2022/46/E/ST2/00115 and 2024/55/B/ST2/01590).

\bibliography{ref}

\begin{thebibliography}{53}%
\makeatletter
\providecommand \@ifxundefined [1]{%
 \@ifx{#1\undefined}
}%
\providecommand \@ifnum [1]{%
 \ifnum #1\expandafter \@firstoftwo
 \else \expandafter \@secondoftwo
 \fi
}%
\providecommand \@ifx [1]{%
 \ifx #1\expandafter \@firstoftwo
 \else \expandafter \@secondoftwo
 \fi
}%
\providecommand \natexlab [1]{#1}%
\providecommand \enquote  [1]{``#1''}%
\providecommand \bibnamefont  [1]{#1}%
\providecommand \bibfnamefont [1]{#1}%
\providecommand \citenamefont [1]{#1}%
\providecommand \href@noop [0]{\@secondoftwo}%
\providecommand \href [0]{\begingroup \@sanitize@url \@href}%
\providecommand \@href[1]{\@@startlink{#1}\@@href}%
\providecommand \@@href[1]{\endgroup#1\@@endlink}%
\providecommand \@sanitize@url [0]{\catcode `\\12\catcode `\$12\catcode `\&12\catcode `\#12\catcode `\^12\catcode `\_12\catcode `\%12\relax}%
\providecommand \@@startlink[1]{}%
\providecommand \@@endlink[0]{}%
\providecommand \url  [0]{\begingroup\@sanitize@url \@url }%
\providecommand \@url [1]{\endgroup\@href {#1}{\urlprefix }}%
\providecommand \urlprefix  [0]{URL }%
\providecommand \Eprint [0]{\href }%
\providecommand \doibase [0]{https://doi.org/}%
\providecommand \selectlanguage [0]{\@gobble}%
\providecommand \bibinfo  [0]{\@secondoftwo}%
\providecommand \bibfield  [0]{\@secondoftwo}%
\providecommand \translation [1]{[#1]}%
\providecommand \BibitemOpen [0]{}%
\providecommand \bibitemStop [0]{}%
\providecommand \bibitemNoStop [0]{.\EOS\space}%
\providecommand \EOS [0]{\spacefactor3000\relax}%
\providecommand \BibitemShut  [1]{\csname bibitem#1\endcsname}%
\let\auto@bib@innerbib\@empty
\bibitem [{\citenamefont {Horodecki}\ \emph {et~al.}(2009)\citenamefont {Horodecki}, \citenamefont {Horodecki}, \citenamefont {Horodecki},\ and\ \citenamefont {Horodecki}}]{Horodecki_2009}%
  \BibitemOpen
  \bibfield  {author} {\bibinfo {author} {\bibfnamefont {R.}~\bibnamefont {Horodecki}}, \bibinfo {author} {\bibfnamefont {P.}~\bibnamefont {Horodecki}}, \bibinfo {author} {\bibfnamefont {M.}~\bibnamefont {Horodecki}},\ and\ \bibinfo {author} {\bibfnamefont {K.}~\bibnamefont {Horodecki}},\ }\bibfield  {title} {\bibinfo {title} {Quantum entanglement},\ }\href {https://doi.org/10.1103/revmodphys.81.865} {\bibfield  {journal} {\bibinfo  {journal} {Reviews of Modern Physics}\ }\textbf {\bibinfo {volume} {81}},\ \bibinfo {pages} {865–942} (\bibinfo {year} {2009})}\BibitemShut {NoStop}%
\bibitem [{\citenamefont {Bennett}\ \emph {et~al.}(1993)\citenamefont {Bennett}, \citenamefont {Brassard}, \citenamefont {Cr\'epeau}, \citenamefont {Jozsa}, \citenamefont {Peres},\ and\ \citenamefont {Wootters}}]{PhysRevLett.70.1895}%
  \BibitemOpen
  \bibfield  {author} {\bibinfo {author} {\bibfnamefont {C.~H.}\ \bibnamefont {Bennett}}, \bibinfo {author} {\bibfnamefont {G.}~\bibnamefont {Brassard}}, \bibinfo {author} {\bibfnamefont {C.}~\bibnamefont {Cr\'epeau}}, \bibinfo {author} {\bibfnamefont {R.}~\bibnamefont {Jozsa}}, \bibinfo {author} {\bibfnamefont {A.}~\bibnamefont {Peres}},\ and\ \bibinfo {author} {\bibfnamefont {W.~K.}\ \bibnamefont {Wootters}},\ }\bibfield  {title} {\bibinfo {title} {Teleporting an unknown quantum state via dual classical and {E}instein-{P}odolsky-{R}osen channels},\ }\href {https://doi.org/10.1103/PhysRevLett.70.1895} {\bibfield  {journal} {\bibinfo  {journal} {Physical Review Letters}\ }\textbf {\bibinfo {volume} {70}},\ \bibinfo {pages} {1895} (\bibinfo {year} {1993})}\BibitemShut {NoStop}%
\bibitem [{\citenamefont {Bennett}\ and\ \citenamefont {Wiesner}(1992)}]{PhysRevLett.69.2881}%
  \BibitemOpen
  \bibfield  {author} {\bibinfo {author} {\bibfnamefont {C.~H.}\ \bibnamefont {Bennett}}\ and\ \bibinfo {author} {\bibfnamefont {S.~J.}\ \bibnamefont {Wiesner}},\ }\bibfield  {title} {\bibinfo {title} {Communication via one- and two-particle operators on {E}instein-{P}odolsky-{R}osen states},\ }\href {https://doi.org/10.1103/PhysRevLett.69.2881} {\bibfield  {journal} {\bibinfo  {journal} {Physical Review Letters}\ }\textbf {\bibinfo {volume} {69}},\ \bibinfo {pages} {2881} (\bibinfo {year} {1992})}\BibitemShut {NoStop}%
\bibitem [{\citenamefont {Shor}(1997)}]{Shor_1997}%
  \BibitemOpen
  \bibfield  {author} {\bibinfo {author} {\bibfnamefont {P.~W.}\ \bibnamefont {Shor}},\ }\bibfield  {title} {\bibinfo {title} {Polynomial-time algorithms for prime factorization and discrete logarithms on a quantum computer},\ }\href {https://doi.org/10.1137/s0097539795293172} {\bibfield  {journal} {\bibinfo  {journal} {SIAM Journal on Computing}\ }\textbf {\bibinfo {volume} {26}},\ \bibinfo {pages} {1484–1509} (\bibinfo {year} {1997})}\BibitemShut {NoStop}%
\bibitem [{\citenamefont {Ekert}(1991)}]{PhysRevLett.67.661}%
  \BibitemOpen
  \bibfield  {author} {\bibinfo {author} {\bibfnamefont {A.~K.}\ \bibnamefont {Ekert}},\ }\bibfield  {title} {\bibinfo {title} {Quantum cryptography based on {B}ell's theorem},\ }\href {https://doi.org/10.1103/PhysRevLett.67.661} {\bibfield  {journal} {\bibinfo  {journal} {Physical Review Letters}\ }\textbf {\bibinfo {volume} {67}},\ \bibinfo {pages} {661} (\bibinfo {year} {1991})}\BibitemShut {NoStop}%
\bibitem [{\citenamefont {Bennett}\ \emph {et~al.}(1999)\citenamefont {Bennett}, \citenamefont {DiVincenzo}, \citenamefont {Fuchs}, \citenamefont {Mor}, \citenamefont {Rains}, \citenamefont {Shor}, \citenamefont {Smolin},\ and\ \citenamefont {Wootters}}]{PhysRevA.59.1070}%
  \BibitemOpen
  \bibfield  {author} {\bibinfo {author} {\bibfnamefont {C.~H.}\ \bibnamefont {Bennett}}, \bibinfo {author} {\bibfnamefont {D.~P.}\ \bibnamefont {DiVincenzo}}, \bibinfo {author} {\bibfnamefont {C.~A.}\ \bibnamefont {Fuchs}}, \bibinfo {author} {\bibfnamefont {T.}~\bibnamefont {Mor}}, \bibinfo {author} {\bibfnamefont {E.}~\bibnamefont {Rains}}, \bibinfo {author} {\bibfnamefont {P.~W.}\ \bibnamefont {Shor}}, \bibinfo {author} {\bibfnamefont {J.~A.}\ \bibnamefont {Smolin}},\ and\ \bibinfo {author} {\bibfnamefont {W.~K.}\ \bibnamefont {Wootters}},\ }\bibfield  {title} {\bibinfo {title} {Quantum nonlocality without entanglement},\ }\href {https://doi.org/10.1103/PhysRevA.59.1070} {\bibfield  {journal} {\bibinfo  {journal} {Physical Review A}\ }\textbf {\bibinfo {volume} {59}},\ \bibinfo {pages} {1070} (\bibinfo {year} {1999})}\BibitemShut {NoStop}%
\bibitem [{\citenamefont {Gross}\ \emph {et~al.}(2009)\citenamefont {Gross}, \citenamefont {Flammia},\ and\ \citenamefont {Eisert}}]{PhysRevLett.102.190501}%
  \BibitemOpen
  \bibfield  {author} {\bibinfo {author} {\bibfnamefont {D.}~\bibnamefont {Gross}}, \bibinfo {author} {\bibfnamefont {S.~T.}\ \bibnamefont {Flammia}},\ and\ \bibinfo {author} {\bibfnamefont {J.}~\bibnamefont {Eisert}},\ }\bibfield  {title} {\bibinfo {title} {Most quantum states are too entangled to be useful as computational resources},\ }\href {https://doi.org/10.1103/PhysRevLett.102.190501} {\bibfield  {journal} {\bibinfo  {journal} {Physical Review Letters}\ }\textbf {\bibinfo {volume} {102}},\ \bibinfo {pages} {190501} (\bibinfo {year} {2009})}\BibitemShut {NoStop}%
\bibitem [{\citenamefont {Peres}\ and\ \citenamefont {Wootters}(1991)}]{PhysRevLett.66.1119}%
  \BibitemOpen
  \bibfield  {author} {\bibinfo {author} {\bibfnamefont {A.}~\bibnamefont {Peres}}\ and\ \bibinfo {author} {\bibfnamefont {W.~K.}\ \bibnamefont {Wootters}},\ }\bibfield  {title} {\bibinfo {title} {Optimal detection of quantum information},\ }\href {https://doi.org/10.1103/PhysRevLett.66.1119} {\bibfield  {journal} {\bibinfo  {journal} {Physical Review Letters}\ }\textbf {\bibinfo {volume} {66}},\ \bibinfo {pages} {1119} (\bibinfo {year} {1991})}\BibitemShut {NoStop}%
\bibitem [{\citenamefont {Walgate}\ \emph {et~al.}(2000)\citenamefont {Walgate}, \citenamefont {Short}, \citenamefont {Hardy},\ and\ \citenamefont {Vedral}}]{PhysRevLett.85.4972}%
  \BibitemOpen
  \bibfield  {author} {\bibinfo {author} {\bibfnamefont {J.}~\bibnamefont {Walgate}}, \bibinfo {author} {\bibfnamefont {A.~J.}\ \bibnamefont {Short}}, \bibinfo {author} {\bibfnamefont {L.}~\bibnamefont {Hardy}},\ and\ \bibinfo {author} {\bibfnamefont {V.}~\bibnamefont {Vedral}},\ }\bibfield  {title} {\bibinfo {title} {Local distinguishability of multipartite orthogonal quantum states},\ }\href {https://doi.org/10.1103/PhysRevLett.85.4972} {\bibfield  {journal} {\bibinfo  {journal} {Physical Review Letters}\ }\textbf {\bibinfo {volume} {85}},\ \bibinfo {pages} {4972} (\bibinfo {year} {2000})}\BibitemShut {NoStop}%
\bibitem [{\citenamefont {Groisman}\ and\ \citenamefont {Vaidman}(2001)}]{Berry_Groisman_2001}%
  \BibitemOpen
  \bibfield  {author} {\bibinfo {author} {\bibfnamefont {B.}~\bibnamefont {Groisman}}\ and\ \bibinfo {author} {\bibfnamefont {L.}~\bibnamefont {Vaidman}},\ }\bibfield  {title} {\bibinfo {title} {Nonlocal variables with product-state eigenstates},\ }\href {https://doi.org/10.1088/0305-4470/34/35/313} {\bibfield  {journal} {\bibinfo  {journal} {Journal of Physics A: Mathematical and General}\ }\textbf {\bibinfo {volume} {34}},\ \bibinfo {pages} {6881} (\bibinfo {year} {2001})}\BibitemShut {NoStop}%
\bibitem [{\citenamefont {Yu}\ \emph {et~al.}(2012)\citenamefont {Yu}, \citenamefont {Duan},\ and\ \citenamefont {Ying}}]{Yu_2012}%
  \BibitemOpen
  \bibfield  {author} {\bibinfo {author} {\bibfnamefont {N.}~\bibnamefont {Yu}}, \bibinfo {author} {\bibfnamefont {R.}~\bibnamefont {Duan}},\ and\ \bibinfo {author} {\bibfnamefont {M.}~\bibnamefont {Ying}},\ }\bibfield  {title} {\bibinfo {title} {Four locally indistinguishable ququad-ququad orthogonal maximally entangled states},\ }\href {https://doi.org/10.1103/PhysRevLett.109.020506} {\bibfield  {journal} {\bibinfo  {journal} {Physical Review Letters}\ }\textbf {\bibinfo {volume} {109}},\ \bibinfo {pages} {020506} (\bibinfo {year} {2012})}\BibitemShut {NoStop}%
\bibitem [{\citenamefont {Yang}\ \emph {et~al.}(2015)\citenamefont {Yang}, \citenamefont {Gao}, \citenamefont {Xu}, \citenamefont {Zuo}, \citenamefont {Zhang},\ and\ \citenamefont {Wen}}]{yang2015characterizing}%
  \BibitemOpen
  \bibfield  {author} {\bibinfo {author} {\bibfnamefont {Y.-H.}\ \bibnamefont {Yang}}, \bibinfo {author} {\bibfnamefont {F.}~\bibnamefont {Gao}}, \bibinfo {author} {\bibfnamefont {G.-B.}\ \bibnamefont {Xu}}, \bibinfo {author} {\bibfnamefont {H.-J.}\ \bibnamefont {Zuo}}, \bibinfo {author} {\bibfnamefont {Z.-C.}\ \bibnamefont {Zhang}},\ and\ \bibinfo {author} {\bibfnamefont {Q.-Y.}\ \bibnamefont {Wen}},\ }\bibfield  {title} {\bibinfo {title} {Characterizing unextendible product bases in qutrit-ququad system},\ }\href {https://doi.org/10.1038/srep11963} {\bibfield  {journal} {\bibinfo  {journal} {Scientific Reports}\ }\textbf {\bibinfo {volume} {5}},\ \bibinfo {pages} {11963} (\bibinfo {year} {2015})}\BibitemShut {NoStop}%
\bibitem [{\citenamefont {Xu}\ \emph {et~al.}(2016{\natexlab{a}})\citenamefont {Xu}, \citenamefont {Yang}, \citenamefont {Wen}, \citenamefont {Qin},\ and\ \citenamefont {Gao}}]{xu2016locally}%
  \BibitemOpen
  \bibfield  {author} {\bibinfo {author} {\bibfnamefont {G.-B.}\ \bibnamefont {Xu}}, \bibinfo {author} {\bibfnamefont {Y.-H.}\ \bibnamefont {Yang}}, \bibinfo {author} {\bibfnamefont {Q.-Y.}\ \bibnamefont {Wen}}, \bibinfo {author} {\bibfnamefont {S.-J.}\ \bibnamefont {Qin}},\ and\ \bibinfo {author} {\bibfnamefont {F.}~\bibnamefont {Gao}},\ }\bibfield  {title} {\bibinfo {title} {Locally indistinguishable orthogonal product bases in arbitrary bipartite quantum system},\ }\href {https://doi.org/10.1038/srep31048} {\bibfield  {journal} {\bibinfo  {journal} {Scientific Reports}\ }\textbf {\bibinfo {volume} {6}},\ \bibinfo {pages} {31048} (\bibinfo {year} {2016}{\natexlab{a}})}\BibitemShut {NoStop}%
\bibitem [{\citenamefont {DiVincenzo}\ \emph {et~al.}(2003)\citenamefont {DiVincenzo}, \citenamefont {Mor}, \citenamefont {Shor}, \citenamefont {Smolin},\ and\ \citenamefont {Terhal}}]{divincenzo2003unextendible}%
  \BibitemOpen
  \bibfield  {author} {\bibinfo {author} {\bibfnamefont {D.~P.}\ \bibnamefont {DiVincenzo}}, \bibinfo {author} {\bibfnamefont {T.}~\bibnamefont {Mor}}, \bibinfo {author} {\bibfnamefont {P.~W.}\ \bibnamefont {Shor}}, \bibinfo {author} {\bibfnamefont {J.~A.}\ \bibnamefont {Smolin}},\ and\ \bibinfo {author} {\bibfnamefont {B.~M.}\ \bibnamefont {Terhal}},\ }\bibfield  {title} {\bibinfo {title} {Unextendible product bases, uncompletable product bases and bound entanglement},\ }\href {https://doi.org/10.1007/s00220-003-0877-6} {\bibfield  {journal} {\bibinfo  {journal} {Communications in Mathematical Physics}\ }\textbf {\bibinfo {volume} {238}},\ \bibinfo {pages} {379} (\bibinfo {year} {2003})}\BibitemShut {NoStop}%
\bibitem [{\citenamefont {Feng}\ and\ \citenamefont {Shi}(2009)}]{4957660}%
  \BibitemOpen
  \bibfield  {author} {\bibinfo {author} {\bibfnamefont {Y.}~\bibnamefont {Feng}}\ and\ \bibinfo {author} {\bibfnamefont {Y.}~\bibnamefont {Shi}},\ }\bibfield  {title} {\bibinfo {title} {Characterizing locally indistinguishable orthogonal product states},\ }\href {https://doi.org/10.1109/TIT.2009.2018330} {\bibfield  {journal} {\bibinfo  {journal} {IEEE Transactions on Information Theory}\ }\textbf {\bibinfo {volume} {55}},\ \bibinfo {pages} {2799} (\bibinfo {year} {2009})}\BibitemShut {NoStop}%
\bibitem [{\citenamefont {Niset}\ and\ \citenamefont {Cerf}(2006)}]{PhysRevA.74.052103}%
  \BibitemOpen
  \bibfield  {author} {\bibinfo {author} {\bibfnamefont {J.}~\bibnamefont {Niset}}\ and\ \bibinfo {author} {\bibfnamefont {N.~J.}\ \bibnamefont {Cerf}},\ }\bibfield  {title} {\bibinfo {title} {Multipartite nonlocality without entanglement in many dimensions},\ }\href {https://doi.org/10.1103/PhysRevA.74.052103} {\bibfield  {journal} {\bibinfo  {journal} {Physical Review A}\ }\textbf {\bibinfo {volume} {74}},\ \bibinfo {pages} {052103} (\bibinfo {year} {2006})}\BibitemShut {NoStop}%
\bibitem [{\citenamefont {Yang}\ \emph {et~al.}(2013)\citenamefont {Yang}, \citenamefont {Gao}, \citenamefont {Tian}, \citenamefont {Cao},\ and\ \citenamefont {Wen}}]{PhysRevA.88.024301}%
  \BibitemOpen
  \bibfield  {author} {\bibinfo {author} {\bibfnamefont {Y.-H.}\ \bibnamefont {Yang}}, \bibinfo {author} {\bibfnamefont {F.}~\bibnamefont {Gao}}, \bibinfo {author} {\bibfnamefont {G.-J.}\ \bibnamefont {Tian}}, \bibinfo {author} {\bibfnamefont {T.-Q.}\ \bibnamefont {Cao}},\ and\ \bibinfo {author} {\bibfnamefont {Q.-Y.}\ \bibnamefont {Wen}},\ }\bibfield  {title} {\bibinfo {title} {Local distinguishability of orthogonal quantum states in a $2\ensuremath{\bigotimes}2\ensuremath{\bigotimes}2$ system},\ }\href {https://doi.org/10.1103/PhysRevA.88.024301} {\bibfield  {journal} {\bibinfo  {journal} {Physical Review A}\ }\textbf {\bibinfo {volume} {88}},\ \bibinfo {pages} {024301} (\bibinfo {year} {2013})}\BibitemShut {NoStop}%
\bibitem [{\citenamefont {Halder}(2018)}]{PhysRevA.98.022303}%
  \BibitemOpen
  \bibfield  {author} {\bibinfo {author} {\bibfnamefont {S.}~\bibnamefont {Halder}},\ }\bibfield  {title} {\bibinfo {title} {Several nonlocal sets of multipartite pure orthogonal product states},\ }\href {https://doi.org/10.1103/PhysRevA.98.022303} {\bibfield  {journal} {\bibinfo  {journal} {Physical Review A}\ }\textbf {\bibinfo {volume} {98}},\ \bibinfo {pages} {022303} (\bibinfo {year} {2018})}\BibitemShut {NoStop}%
\bibitem [{\citenamefont {Xu}\ \emph {et~al.}(2016{\natexlab{b}})\citenamefont {Xu}, \citenamefont {Wen}, \citenamefont {Qin}, \citenamefont {Yang},\ and\ \citenamefont {Gao}}]{PhysRevA.93.032341}%
  \BibitemOpen
  \bibfield  {author} {\bibinfo {author} {\bibfnamefont {G.-B.}\ \bibnamefont {Xu}}, \bibinfo {author} {\bibfnamefont {Q.-Y.}\ \bibnamefont {Wen}}, \bibinfo {author} {\bibfnamefont {S.-J.}\ \bibnamefont {Qin}}, \bibinfo {author} {\bibfnamefont {Y.-H.}\ \bibnamefont {Yang}},\ and\ \bibinfo {author} {\bibfnamefont {F.}~\bibnamefont {Gao}},\ }\bibfield  {title} {\bibinfo {title} {Quantum nonlocality of multipartite orthogonal product states},\ }\href {https://doi.org/10.1103/PhysRevA.93.032341} {\bibfield  {journal} {\bibinfo  {journal} {Physical Review A}\ }\textbf {\bibinfo {volume} {93}},\ \bibinfo {pages} {032341} (\bibinfo {year} {2016}{\natexlab{b}})}\BibitemShut {NoStop}%
\bibitem [{\citenamefont {Zhang}\ \emph {et~al.}(2017)\citenamefont {Zhang}, \citenamefont {Zhang}, \citenamefont {Gao}, \citenamefont {Wen},\ and\ \citenamefont {Oh}}]{PhysRevA.95.052344}%
  \BibitemOpen
  \bibfield  {author} {\bibinfo {author} {\bibfnamefont {Z.-C.}\ \bibnamefont {Zhang}}, \bibinfo {author} {\bibfnamefont {K.-J.}\ \bibnamefont {Zhang}}, \bibinfo {author} {\bibfnamefont {F.}~\bibnamefont {Gao}}, \bibinfo {author} {\bibfnamefont {Q.-Y.}\ \bibnamefont {Wen}},\ and\ \bibinfo {author} {\bibfnamefont {C.~H.}\ \bibnamefont {Oh}},\ }\bibfield  {title} {\bibinfo {title} {Construction of nonlocal multipartite quantum states},\ }\href {https://doi.org/10.1103/PhysRevA.95.052344} {\bibfield  {journal} {\bibinfo  {journal} {Physical Review A}\ }\textbf {\bibinfo {volume} {95}},\ \bibinfo {pages} {052344} (\bibinfo {year} {2017})}\BibitemShut {NoStop}%
\bibitem [{\citenamefont {Sen}\ \emph {et~al.}(2022)\citenamefont {Sen}, \citenamefont {Lobo}, \citenamefont {Naik}, \citenamefont {Patra}, \citenamefont {Gupta}, \citenamefont {Ghosh}, \citenamefont {Saha}, \citenamefont {Alimuddin}, \citenamefont {Guha}, \citenamefont {Bhattacharya},\ and\ \citenamefont {Banik}}]{PhysRevA.105.032407}%
  \BibitemOpen
  \bibfield  {author} {\bibinfo {author} {\bibfnamefont {S.}~\bibnamefont {Sen}}, \bibinfo {author} {\bibfnamefont {E.~P.}\ \bibnamefont {Lobo}}, \bibinfo {author} {\bibfnamefont {S.~G.}\ \bibnamefont {Naik}}, \bibinfo {author} {\bibfnamefont {R.~K.}\ \bibnamefont {Patra}}, \bibinfo {author} {\bibfnamefont {T.}~\bibnamefont {Gupta}}, \bibinfo {author} {\bibfnamefont {S.~B.}\ \bibnamefont {Ghosh}}, \bibinfo {author} {\bibfnamefont {S.}~\bibnamefont {Saha}}, \bibinfo {author} {\bibfnamefont {M.}~\bibnamefont {Alimuddin}}, \bibinfo {author} {\bibfnamefont {T.}~\bibnamefont {Guha}}, \bibinfo {author} {\bibfnamefont {S.~S.}\ \bibnamefont {Bhattacharya}},\ and\ \bibinfo {author} {\bibfnamefont {M.}~\bibnamefont {Banik}},\ }\bibfield  {title} {\bibinfo {title} {Local quantum state marking},\ }\href {https://doi.org/10.1103/PhysRevA.105.032407} {\bibfield  {journal} {\bibinfo  {journal} {Physical Review A}\ }\textbf {\bibinfo {volume} {105}},\ \bibinfo {pages} {032407} (\bibinfo {year} {2022})}\BibitemShut {NoStop}%
\bibitem [{\citenamefont {Terhal}\ \emph {et~al.}(2001)\citenamefont {Terhal}, \citenamefont {DiVincenzo},\ and\ \citenamefont {Leung}}]{PhysRevLett.86.5807}%
  \BibitemOpen
  \bibfield  {author} {\bibinfo {author} {\bibfnamefont {B.~M.}\ \bibnamefont {Terhal}}, \bibinfo {author} {\bibfnamefont {D.~P.}\ \bibnamefont {DiVincenzo}},\ and\ \bibinfo {author} {\bibfnamefont {D.~W.}\ \bibnamefont {Leung}},\ }\bibfield  {title} {\bibinfo {title} {Hiding bits in {B}ell states},\ }\href {https://doi.org/10.1103/PhysRevLett.86.5807} {\bibfield  {journal} {\bibinfo  {journal} {Physical Review Letters}\ }\textbf {\bibinfo {volume} {86}},\ \bibinfo {pages} {5807} (\bibinfo {year} {2001})}\BibitemShut {NoStop}%
\bibitem [{\citenamefont {Eggeling}\ and\ \citenamefont {Werner}(2002)}]{PhysRevLett.89.097905}%
  \BibitemOpen
  \bibfield  {author} {\bibinfo {author} {\bibfnamefont {T.}~\bibnamefont {Eggeling}}\ and\ \bibinfo {author} {\bibfnamefont {R.~F.}\ \bibnamefont {Werner}},\ }\bibfield  {title} {\bibinfo {title} {Hiding classical data in multipartite quantum states},\ }\href {https://doi.org/10.1103/PhysRevLett.89.097905} {\bibfield  {journal} {\bibinfo  {journal} {Physical Review Letters}\ }\textbf {\bibinfo {volume} {89}},\ \bibinfo {pages} {097905} (\bibinfo {year} {2002})}\BibitemShut {NoStop}%
\bibitem [{\citenamefont {DiVincenzo}\ \emph {et~al.}(2002)\citenamefont {DiVincenzo}, \citenamefont {Leung},\ and\ \citenamefont {Terhal}}]{DiVincenzo_2002}%
  \BibitemOpen
  \bibfield  {author} {\bibinfo {author} {\bibfnamefont {D.}~\bibnamefont {DiVincenzo}}, \bibinfo {author} {\bibfnamefont {D.}~\bibnamefont {Leung}},\ and\ \bibinfo {author} {\bibfnamefont {B.}~\bibnamefont {Terhal}},\ }\bibfield  {title} {\bibinfo {title} {Quantum data hiding},\ }\href {https://doi.org/10.1109/18.985948} {\bibfield  {journal} {\bibinfo  {journal} {IEEE Transactions on Information Theory}\ }\textbf {\bibinfo {volume} {48}},\ \bibinfo {pages} {580–598} (\bibinfo {year} {2002})}\BibitemShut {NoStop}%
\bibitem [{\citenamefont {Hayden}\ \emph {et~al.}(2004)\citenamefont {Hayden}, \citenamefont {Leung}, \citenamefont {Shor},\ and\ \citenamefont {Winter}}]{Hayden_2004}%
  \BibitemOpen
  \bibfield  {author} {\bibinfo {author} {\bibfnamefont {P.}~\bibnamefont {Hayden}}, \bibinfo {author} {\bibfnamefont {D.}~\bibnamefont {Leung}}, \bibinfo {author} {\bibfnamefont {P.~W.}\ \bibnamefont {Shor}},\ and\ \bibinfo {author} {\bibfnamefont {A.}~\bibnamefont {Winter}},\ }\bibfield  {title} {\bibinfo {title} {Randomizing quantum states: Constructions and applications},\ }\href {https://doi.org/10.1007/s00220-004-1087-6} {\bibfield  {journal} {\bibinfo  {journal} {Communications in Mathematical Physics}\ }\textbf {\bibinfo {volume} {250}},\ \bibinfo {pages} {371–391} (\bibinfo {year} {2004})}\BibitemShut {NoStop}%
\bibitem [{\citenamefont {Aubrun}\ and\ \citenamefont {Lancien}(2015)}]{Aubrun_2015}%
  \BibitemOpen
  \bibfield  {author} {\bibinfo {author} {\bibfnamefont {G.}~\bibnamefont {Aubrun}}\ and\ \bibinfo {author} {\bibfnamefont {C.}~\bibnamefont {Lancien}},\ }\bibfield  {title} {\bibinfo {title} {Locally restricted measurements on a multipartite quantum system: data hiding is generic},\ }\href {https://doi.org/10.26421/qic15.5-6} {\bibfield  {journal} {\bibinfo  {journal} {Quantum Information and Computation}\ }\textbf {\bibinfo {volume} {15}},\ \bibinfo {pages} {513} (\bibinfo {year} {2015})}\BibitemShut {NoStop}%
\bibitem [{\citenamefont {Jonathan}\ and\ \citenamefont {Plenio}(1999)}]{PhysRevLett.83.3566}%
  \BibitemOpen
  \bibfield  {author} {\bibinfo {author} {\bibfnamefont {D.}~\bibnamefont {Jonathan}}\ and\ \bibinfo {author} {\bibfnamefont {M.~B.}\ \bibnamefont {Plenio}},\ }\bibfield  {title} {\bibinfo {title} {Entanglement-assisted local manipulation of pure quantum states},\ }\href {https://doi.org/10.1103/PhysRevLett.83.3566} {\bibfield  {journal} {\bibinfo  {journal} {Physical Review Letters}\ }\textbf {\bibinfo {volume} {83}},\ \bibinfo {pages} {3566} (\bibinfo {year} {1999})}\BibitemShut {NoStop}%
\bibitem [{\citenamefont {Neven}\ \emph {et~al.}(2021)\citenamefont {Neven}, \citenamefont {Gunn}, \citenamefont {Hebenstreit},\ and\ \citenamefont {Kraus}}]{Neven_2021}%
  \BibitemOpen
  \bibfield  {author} {\bibinfo {author} {\bibfnamefont {A.}~\bibnamefont {Neven}}, \bibinfo {author} {\bibfnamefont {D.~K.}\ \bibnamefont {Gunn}}, \bibinfo {author} {\bibfnamefont {M.}~\bibnamefont {Hebenstreit}},\ and\ \bibinfo {author} {\bibfnamefont {B.}~\bibnamefont {Kraus}},\ }\bibfield  {title} {\bibinfo {title} {Local transformations of multiple multipartite states},\ }\href {https://doi.org/10.21468/scipostphys.11.2.042} {\bibfield  {journal} {\bibinfo  {journal} {SciPost Physics}\ }\textbf {\bibinfo {volume} {11}},\ \bibinfo {pages} {042} (\bibinfo {year} {2021})}\BibitemShut {NoStop}%
\bibitem [{\citenamefont {Kondra}\ \emph {et~al.}(2021)\citenamefont {Kondra}, \citenamefont {Datta},\ and\ \citenamefont {Streltsov}}]{PhysRevLett.127.150503}%
  \BibitemOpen
  \bibfield  {author} {\bibinfo {author} {\bibfnamefont {T.~V.}\ \bibnamefont {Kondra}}, \bibinfo {author} {\bibfnamefont {C.}~\bibnamefont {Datta}},\ and\ \bibinfo {author} {\bibfnamefont {A.}~\bibnamefont {Streltsov}},\ }\bibfield  {title} {\bibinfo {title} {Catalytic transformations of pure entangled states},\ }\href {https://doi.org/10.1103/PhysRevLett.127.150503} {\bibfield  {journal} {\bibinfo  {journal} {Physical Review Letters}\ }\textbf {\bibinfo {volume} {127}},\ \bibinfo {pages} {150503} (\bibinfo {year} {2021})}\BibitemShut {NoStop}%
\bibitem [{\citenamefont {Lipka-Bartosik}\ and\ \citenamefont {Skrzypczyk}(2021)}]{PhysRevLett.127.080502}%
  \BibitemOpen
  \bibfield  {author} {\bibinfo {author} {\bibfnamefont {P.}~\bibnamefont {Lipka-Bartosik}}\ and\ \bibinfo {author} {\bibfnamefont {P.}~\bibnamefont {Skrzypczyk}},\ }\bibfield  {title} {\bibinfo {title} {Catalytic quantum teleportation},\ }\href {https://doi.org/10.1103/PhysRevLett.127.080502} {\bibfield  {journal} {\bibinfo  {journal} {Physical Review Letters}\ }\textbf {\bibinfo {volume} {127}},\ \bibinfo {pages} {080502} (\bibinfo {year} {2021})}\BibitemShut {NoStop}%
\bibitem [{\citenamefont {Shiraishi}\ and\ \citenamefont {Sagawa}(2021)}]{PhysRevLett.126.150502}%
  \BibitemOpen
  \bibfield  {author} {\bibinfo {author} {\bibfnamefont {N.}~\bibnamefont {Shiraishi}}\ and\ \bibinfo {author} {\bibfnamefont {T.}~\bibnamefont {Sagawa}},\ }\bibfield  {title} {\bibinfo {title} {Quantum thermodynamics of correlated-catalytic state conversion at small scale},\ }\href {https://doi.org/10.1103/PhysRevLett.126.150502} {\bibfield  {journal} {\bibinfo  {journal} {Physical Review Letters}\ }\textbf {\bibinfo {volume} {126}},\ \bibinfo {pages} {150502} (\bibinfo {year} {2021})}\BibitemShut {NoStop}%
\bibitem [{\citenamefont {Kondra}\ \emph {et~al.}(2024)\citenamefont {Kondra}, \citenamefont {Ganardi},\ and\ \citenamefont {Streltsov}}]{PhysRevLett.132.200201}%
  \BibitemOpen
  \bibfield  {author} {\bibinfo {author} {\bibfnamefont {T.~V.}\ \bibnamefont {Kondra}}, \bibinfo {author} {\bibfnamefont {R.}~\bibnamefont {Ganardi}},\ and\ \bibinfo {author} {\bibfnamefont {A.}~\bibnamefont {Streltsov}},\ }\bibfield  {title} {\bibinfo {title} {Coherence manipulation in asymmetry and thermodynamics},\ }\href {https://doi.org/10.1103/PhysRevLett.132.200201} {\bibfield  {journal} {\bibinfo  {journal} {Physical Review Letters}\ }\textbf {\bibinfo {volume} {132}},\ \bibinfo {pages} {200201} (\bibinfo {year} {2024})}\BibitemShut {NoStop}%
\bibitem [{\citenamefont {Shiraishi}\ and\ \citenamefont {Takagi}(2024)}]{PhysRevLett.132.180202}%
  \BibitemOpen
  \bibfield  {author} {\bibinfo {author} {\bibfnamefont {N.}~\bibnamefont {Shiraishi}}\ and\ \bibinfo {author} {\bibfnamefont {R.}~\bibnamefont {Takagi}},\ }\bibfield  {title} {\bibinfo {title} {Arbitrary amplification of quantum coherence in asymptotic and catalytic transformation},\ }\href {https://doi.org/10.1103/PhysRevLett.132.180202} {\bibfield  {journal} {\bibinfo  {journal} {Physical Review Letters}\ }\textbf {\bibinfo {volume} {132}},\ \bibinfo {pages} {180202} (\bibinfo {year} {2024})}\BibitemShut {NoStop}%
\bibitem [{\citenamefont {Datta}\ \emph {et~al.}(2023)\citenamefont {Datta}, \citenamefont {Kondra}, \citenamefont {Miller},\ and\ \citenamefont {Streltsov}}]{Datta_2023}%
  \BibitemOpen
  \bibfield  {author} {\bibinfo {author} {\bibfnamefont {C.}~\bibnamefont {Datta}}, \bibinfo {author} {\bibfnamefont {T.~V.}\ \bibnamefont {Kondra}}, \bibinfo {author} {\bibfnamefont {M.}~\bibnamefont {Miller}},\ and\ \bibinfo {author} {\bibfnamefont {A.}~\bibnamefont {Streltsov}},\ }\bibfield  {title} {\bibinfo {title} {Catalysis of entanglement and other quantum resources},\ }\href {https://doi.org/10.1088/1361-6633/acfbec} {\bibfield  {journal} {\bibinfo  {journal} {Reports on Progress in Physics}\ }\textbf {\bibinfo {volume} {86}},\ \bibinfo {pages} {116002} (\bibinfo {year} {2023})}\BibitemShut {NoStop}%
\bibitem [{\citenamefont {Lipka-Bartosik}\ \emph {et~al.}(2024)\citenamefont {Lipka-Bartosik}, \citenamefont {Wilming},\ and\ \citenamefont {Ng}}]{RevModPhys.96.025005}%
  \BibitemOpen
  \bibfield  {author} {\bibinfo {author} {\bibfnamefont {P.}~\bibnamefont {Lipka-Bartosik}}, \bibinfo {author} {\bibfnamefont {H.}~\bibnamefont {Wilming}},\ and\ \bibinfo {author} {\bibfnamefont {N.~H.~Y.}\ \bibnamefont {Ng}},\ }\bibfield  {title} {\bibinfo {title} {Catalysis in quantum information theory},\ }\href {https://doi.org/10.1103/RevModPhys.96.025005} {\bibfield  {journal} {\bibinfo  {journal} {Reviews of Modern Physics}\ }\textbf {\bibinfo {volume} {96}},\ \bibinfo {pages} {025005} (\bibinfo {year} {2024})}\BibitemShut {NoStop}%
\bibitem [{\citenamefont {Konig}\ \emph {et~al.}(2005)\citenamefont {Konig}, \citenamefont {Maurer},\ and\ \citenamefont {Renner}}]{Konig_2005}%
  \BibitemOpen
  \bibfield  {author} {\bibinfo {author} {\bibfnamefont {R.}~\bibnamefont {Konig}}, \bibinfo {author} {\bibfnamefont {U.}~\bibnamefont {Maurer}},\ and\ \bibinfo {author} {\bibfnamefont {R.}~\bibnamefont {Renner}},\ }\bibfield  {title} {\bibinfo {title} {On the power of quantum memory},\ }\href {https://doi.org/10.1109/tit.2005.850087} {\bibfield  {journal} {\bibinfo  {journal} {IEEE Transactions on Information Theory}\ }\textbf {\bibinfo {volume} {51}},\ \bibinfo {pages} {2391–2401} (\bibinfo {year} {2005})}\BibitemShut {NoStop}%
\bibitem [{\citenamefont {Berta}\ \emph {et~al.}(2010)\citenamefont {Berta}, \citenamefont {Christandl}, \citenamefont {Colbeck}, \citenamefont {Renes},\ and\ \citenamefont {Renner}}]{Berta_2010}%
  \BibitemOpen
  \bibfield  {author} {\bibinfo {author} {\bibfnamefont {M.}~\bibnamefont {Berta}}, \bibinfo {author} {\bibfnamefont {M.}~\bibnamefont {Christandl}}, \bibinfo {author} {\bibfnamefont {R.}~\bibnamefont {Colbeck}}, \bibinfo {author} {\bibfnamefont {J.~M.}\ \bibnamefont {Renes}},\ and\ \bibinfo {author} {\bibfnamefont {R.}~\bibnamefont {Renner}},\ }\bibfield  {title} {\bibinfo {title} {The uncertainty principle in the presence of quantum memory},\ }\href {https://doi.org/10.1038/nphys1734} {\bibfield  {journal} {\bibinfo  {journal} {Nature Physics}\ }\textbf {\bibinfo {volume} {6}},\ \bibinfo {pages} {659–662} (\bibinfo {year} {2010})}\BibitemShut {NoStop}%
\bibitem [{\citenamefont {Chiribella}\ \emph {et~al.}(2009)\citenamefont {Chiribella}, \citenamefont {D'Ariano},\ and\ \citenamefont {Perinotti}}]{PhysRevA.80.022339}%
  \BibitemOpen
  \bibfield  {author} {\bibinfo {author} {\bibfnamefont {G.}~\bibnamefont {Chiribella}}, \bibinfo {author} {\bibfnamefont {G.~M.}\ \bibnamefont {D'Ariano}},\ and\ \bibinfo {author} {\bibfnamefont {P.}~\bibnamefont {Perinotti}},\ }\bibfield  {title} {\bibinfo {title} {Theoretical framework for quantum networks},\ }\href {https://doi.org/10.1103/PhysRevA.80.022339} {\bibfield  {journal} {\bibinfo  {journal} {Physical Review A}\ }\textbf {\bibinfo {volume} {80}},\ \bibinfo {pages} {022339} (\bibinfo {year} {2009})}\BibitemShut {NoStop}%
\bibitem [{\citenamefont {Gutoski}\ and\ \citenamefont {Watrous}(2007)}]{10.1145/1250790.1250873}%
  \BibitemOpen
  \bibfield  {author} {\bibinfo {author} {\bibfnamefont {G.}~\bibnamefont {Gutoski}}\ and\ \bibinfo {author} {\bibfnamefont {J.}~\bibnamefont {Watrous}},\ }\bibfield  {title} {\bibinfo {title} {Toward a general theory of quantum games},\ }in\ \href {https://doi.org/10.1145/1250790.1250873} {\emph {\bibinfo {booktitle} {Proceedings of the Thirty-Ninth Annual ACM Symposium on Theory of Computing}}},\ \bibinfo {series and number} {STOC '07}\ (\bibinfo  {publisher} {Association for Computing Machinery},\ \bibinfo {address} {New York, NY, USA},\ \bibinfo {year} {2007})\ p.\ \bibinfo {pages} {565–574}\BibitemShut {NoStop}%
\bibitem [{\citenamefont {Harrow}\ \emph {et~al.}(2010)\citenamefont {Harrow}, \citenamefont {Hassidim}, \citenamefont {Leung},\ and\ \citenamefont {Watrous}}]{PhysRevA.81.032339}%
  \BibitemOpen
  \bibfield  {author} {\bibinfo {author} {\bibfnamefont {A.~W.}\ \bibnamefont {Harrow}}, \bibinfo {author} {\bibfnamefont {A.}~\bibnamefont {Hassidim}}, \bibinfo {author} {\bibfnamefont {D.~W.}\ \bibnamefont {Leung}},\ and\ \bibinfo {author} {\bibfnamefont {J.}~\bibnamefont {Watrous}},\ }\bibfield  {title} {\bibinfo {title} {Adaptive versus nonadaptive strategies for quantum channel discrimination},\ }\href {https://doi.org/10.1103/PhysRevA.81.032339} {\bibfield  {journal} {\bibinfo  {journal} {Physical Review A}\ }\textbf {\bibinfo {volume} {81}},\ \bibinfo {pages} {032339} (\bibinfo {year} {2010})}\BibitemShut {NoStop}%
\bibitem [{\citenamefont {Chiribella}\ \emph {et~al.}(2008)\citenamefont {Chiribella}, \citenamefont {D'Ariano},\ and\ \citenamefont {Perinotti}}]{PhysRevLett.101.180501}%
  \BibitemOpen
  \bibfield  {author} {\bibinfo {author} {\bibfnamefont {G.}~\bibnamefont {Chiribella}}, \bibinfo {author} {\bibfnamefont {G.~M.}\ \bibnamefont {D'Ariano}},\ and\ \bibinfo {author} {\bibfnamefont {P.}~\bibnamefont {Perinotti}},\ }\bibfield  {title} {\bibinfo {title} {Memory effects in quantum channel discrimination},\ }\href {https://doi.org/10.1103/PhysRevLett.101.180501} {\bibfield  {journal} {\bibinfo  {journal} {Physical Review Letters}\ }\textbf {\bibinfo {volume} {101}},\ \bibinfo {pages} {180501} (\bibinfo {year} {2008})}\BibitemShut {NoStop}%
\bibitem [{\citenamefont {Helstrom}(1969)}]{helstrom1969quantum}%
  \BibitemOpen
  \bibfield  {author} {\bibinfo {author} {\bibfnamefont {C.~W.}\ \bibnamefont {Helstrom}},\ }\bibfield  {title} {\bibinfo {title} {Quantum detection and estimation theory},\ }\href {https://doi.org/10.1007/BF01007479} {\bibfield  {journal} {\bibinfo  {journal} {Journal of Statistical Physics}\ }\textbf {\bibinfo {volume} {1}},\ \bibinfo {pages} {231} (\bibinfo {year} {1969})}\BibitemShut {NoStop}%
\bibitem [{\citenamefont {Helstrom}(1967)}]{HELSTROM1967254}%
  \BibitemOpen
  \bibfield  {author} {\bibinfo {author} {\bibfnamefont {C.~W.}\ \bibnamefont {Helstrom}},\ }\bibfield  {title} {\bibinfo {title} {Detection theory and quantum mechanics},\ }\href {https://doi.org/https://doi.org/10.1016/S0019-9958(67)90302-6} {\bibfield  {journal} {\bibinfo  {journal} {Information and Control}\ }\textbf {\bibinfo {volume} {10}},\ \bibinfo {pages} {254} (\bibinfo {year} {1967})}\BibitemShut {NoStop}%
\bibitem [{\citenamefont {Matthews}\ \emph {et~al.}(2009)\citenamefont {Matthews}, \citenamefont {Wehner},\ and\ \citenamefont {Winter}}]{Matthews_2009}%
  \BibitemOpen
  \bibfield  {author} {\bibinfo {author} {\bibfnamefont {W.}~\bibnamefont {Matthews}}, \bibinfo {author} {\bibfnamefont {S.}~\bibnamefont {Wehner}},\ and\ \bibinfo {author} {\bibfnamefont {A.}~\bibnamefont {Winter}},\ }\bibfield  {title} {\bibinfo {title} {Distinguishability of quantum states under restricted families of measurements with an application to quantum data hiding},\ }\href {https://doi.org/10.1007/s00220-009-0890-5} {\bibfield  {journal} {\bibinfo  {journal} {Communications in Mathematical Physics}\ }\textbf {\bibinfo {volume} {291}},\ \bibinfo {pages} {813–843} (\bibinfo {year} {2009})}\BibitemShut {NoStop}%
\bibitem [{\citenamefont {Ha}\ and\ \citenamefont {Kim}(2025)}]{Ha_2025}%
  \BibitemOpen
  \bibfield  {author} {\bibinfo {author} {\bibfnamefont {D.}~\bibnamefont {Ha}}\ and\ \bibinfo {author} {\bibfnamefont {J.~S.}\ \bibnamefont {Kim}},\ }\bibfield  {title} {\bibinfo {title} {Quantum data-hiding scheme using orthogonal separable states},\ }\href {https://doi.org/10.1103/PhysRevA.111.052405} {\bibfield  {journal} {\bibinfo  {journal} {Physical Review A}\ }\textbf {\bibinfo {volume} {111}},\ \bibinfo {pages} {052405} (\bibinfo {year} {2025})}\BibitemShut {NoStop}%
\bibitem [{\citenamefont {Mele}\ and\ \citenamefont {Lami}(2025)}]{mele2025optimisingquantumdatahiding}%
  \BibitemOpen
  \bibfield  {author} {\bibinfo {author} {\bibfnamefont {F.~A.}\ \bibnamefont {Mele}}\ and\ \bibinfo {author} {\bibfnamefont {L.}~\bibnamefont {Lami}},\ }\href {https://arxiv.org/abs/2510.03538} {\bibinfo {title} {Optimising quantum data hiding}} (\bibinfo {year} {2025}),\ \Eprint {https://arxiv.org/abs/2510.03538} {arXiv:2510.03538 [quant-ph]} \BibitemShut {NoStop}%
\bibitem [{\citenamefont {Werner}(1989)}]{PhysRevA.40.4277}%
  \BibitemOpen
  \bibfield  {author} {\bibinfo {author} {\bibfnamefont {R.~F.}\ \bibnamefont {Werner}},\ }\bibfield  {title} {\bibinfo {title} {Quantum states with {E}instein-{P}odolsky-{R}osen correlations admitting a hidden-variable model},\ }\href {https://doi.org/10.1103/PhysRevA.40.4277} {\bibfield  {journal} {\bibinfo  {journal} {Physical Review A}\ }\textbf {\bibinfo {volume} {40}},\ \bibinfo {pages} {4277} (\bibinfo {year} {1989})}\BibitemShut {NoStop}%
\bibitem [{\citenamefont {Vedral}\ \emph {et~al.}(1997)\citenamefont {Vedral}, \citenamefont {Plenio}, \citenamefont {Rippin},\ and\ \citenamefont {Knight}}]{PhysRevLett.78.2275}%
  \BibitemOpen
  \bibfield  {author} {\bibinfo {author} {\bibfnamefont {V.}~\bibnamefont {Vedral}}, \bibinfo {author} {\bibfnamefont {M.~B.}\ \bibnamefont {Plenio}}, \bibinfo {author} {\bibfnamefont {M.~A.}\ \bibnamefont {Rippin}},\ and\ \bibinfo {author} {\bibfnamefont {P.~L.}\ \bibnamefont {Knight}},\ }\bibfield  {title} {\bibinfo {title} {{Quantifying Entanglement}},\ }\href {https://doi.org/10.1103/PhysRevLett.78.2275} {\bibfield  {journal} {\bibinfo  {journal} {Physical Review Letters}\ }\textbf {\bibinfo {volume} {78}},\ \bibinfo {pages} {2275} (\bibinfo {year} {1997})}\BibitemShut {NoStop}%
\bibitem [{\citenamefont {Donald}\ \emph {et~al.}(2002)\citenamefont {Donald}, \citenamefont {Horodecki},\ and\ \citenamefont {Rudolph}}]{Donald_2002}%
  \BibitemOpen
  \bibfield  {author} {\bibinfo {author} {\bibfnamefont {M.~J.}\ \bibnamefont {Donald}}, \bibinfo {author} {\bibfnamefont {M.}~\bibnamefont {Horodecki}},\ and\ \bibinfo {author} {\bibfnamefont {O.}~\bibnamefont {Rudolph}},\ }\bibfield  {title} {\bibinfo {title} {The uniqueness theorem for entanglement measures},\ }\href {https://doi.org/10.1063/1.1495917} {\bibfield  {journal} {\bibinfo  {journal} {Journal of Mathematical Physics}\ }\textbf {\bibinfo {volume} {43}},\ \bibinfo {pages} {4252–4272} (\bibinfo {year} {2002})}\BibitemShut {NoStop}%
\bibitem [{\citenamefont {Chitambar}\ \emph {et~al.}(2014)\citenamefont {Chitambar}, \citenamefont {Leung}, \citenamefont {Man{\v{c}}inska}, \citenamefont {Ozols},\ and\ \citenamefont {Winter}}]{Chitambar2014}%
  \BibitemOpen
  \bibfield  {author} {\bibinfo {author} {\bibfnamefont {E.}~\bibnamefont {Chitambar}}, \bibinfo {author} {\bibfnamefont {D.}~\bibnamefont {Leung}}, \bibinfo {author} {\bibfnamefont {L.}~\bibnamefont {Man{\v{c}}inska}}, \bibinfo {author} {\bibfnamefont {M.}~\bibnamefont {Ozols}},\ and\ \bibinfo {author} {\bibfnamefont {A.}~\bibnamefont {Winter}},\ }\bibfield  {title} {\bibinfo {title} {{Everything You Always Wanted to Know About LOCC (But Were Afraid to Ask)}},\ }\href {https://doi.org/10.1007/s00220-014-1953-9} {\bibfield  {journal} {\bibinfo  {journal} {Communications in Mathematical Physics}\ }\textbf {\bibinfo {volume} {328}},\ \bibinfo {pages} {303} (\bibinfo {year} {2014})}\BibitemShut {NoStop}%
\bibitem [{\citenamefont {Hoeffding}(1963)}]{H63}%
  \BibitemOpen
  \bibfield  {author} {\bibinfo {author} {\bibfnamefont {W.}~\bibnamefont {Hoeffding}},\ }\bibfield  {title} {\bibinfo {title} {Probability inequalities for sums of bounded random variables},\ }\href {https://doi.org/10.2307/2282952} {\bibfield  {journal} {\bibinfo  {journal} {Journal of the American Statistical Association}\ }\textbf {\bibinfo {volume} {58}},\ \bibinfo {pages} {13} (\bibinfo {year} {1963})}\BibitemShut {NoStop}%
\bibitem [{\citenamefont {Azuma}(1967)}]{azuma1967weighted}%
  \BibitemOpen
  \bibfield  {author} {\bibinfo {author} {\bibfnamefont {K.}~\bibnamefont {Azuma}},\ }\bibfield  {title} {\bibinfo {title} {Weighted sums of certain dependent random variables},\ }\href {https://doi.org/10.2748/tmj/1178243286} {\bibfield  {journal} {\bibinfo  {journal} {Tohoku Mathematical Journal, Second Series}\ }\textbf {\bibinfo {volume} {19}},\ \bibinfo {pages} {357} (\bibinfo {year} {1967})}\BibitemShut {NoStop}%
\bibitem [{\citenamefont {Hayashi}\ \emph {et~al.}(2002)\citenamefont {Hayashi}, \citenamefont {Koashi}, \citenamefont {Matsumoto}, \citenamefont {Morikoshi},\ and\ \citenamefont {Winter}}]{Hayashi_2003}%
  \BibitemOpen
  \bibfield  {author} {\bibinfo {author} {\bibfnamefont {M.}~\bibnamefont {Hayashi}}, \bibinfo {author} {\bibfnamefont {M.}~\bibnamefont {Koashi}}, \bibinfo {author} {\bibfnamefont {K.}~\bibnamefont {Matsumoto}}, \bibinfo {author} {\bibfnamefont {F.}~\bibnamefont {Morikoshi}},\ and\ \bibinfo {author} {\bibfnamefont {A.}~\bibnamefont {Winter}},\ }\bibfield  {title} {\bibinfo {title} {Error exponents for entanglement concentration},\ }\href {https://doi.org/10.1088/0305-4470/36/2/316} {\bibfield  {journal} {\bibinfo  {journal} {Journal of Physics A: Mathematical and General}\ }\textbf {\bibinfo {volume} {36}},\ \bibinfo {pages} {527} (\bibinfo {year} {2002})}\BibitemShut {NoStop}%
\end{thebibliography}%

\end{document}